\documentclass[preprint,12pt]{elsarticle}

\usepackage{amssymb}
\usepackage{amsmath}
\usepackage{bm}
\usepackage{cases}
\usepackage{xcolor,graphicx,float}
\usepackage{color, soul}
\usepackage{stfloats}
\usepackage[numbers,sort&compress]{natbib}
\usepackage[amsmath,thmmarks]{ntheorem}
\usepackage{theorem}
\usepackage{algorithm}
\usepackage{algorithmic}
\usepackage{graphicx}
\usepackage{subfigure}
\usepackage{pict2e}

\newtheorem{proposition}{Proposition}

\usepackage{etoolbox}

\theoremheaderfont{\sc}\theorembodyfont{\upshape}
\theoremstyle{nonumberplain}
\theoremseparator{}
\theoremsymbol{\rule{1ex}{1ex}}
\newtheorem{proof}{Proof}

\biboptions{sort&compress}
\journal{Signal Processing}

\begin{document}

\begin{frontmatter}



\title{Fast and Efficient Implementation of the Maximum Likelihood Estimation for the Linear Regression with Gaussian Model Uncertainty}


\author{Ruohai Guo, Jiang~Zhu, Xing Jiang and Fengzhong Qu} %% Author name

\affiliation{organization={the State Key Laboratory of Ocean Sensing, Zhejiang University},%Department and Organization
            addressline={No.1 Zheda Road}, 
            city={Zhoushan},
            postcode={316021}, 
            country={China}}

%% Abstract
\begin{abstract}

The linear regression model with a random variable (RV) measurement matrix, where the mean of the  random measurement matrix has full column rank, has been extensively studied. In particular, the quasiconvexity of the maximum likelihood estimation (MLE) problem was established, and the corresponding Cram\'{e}r–Rao bound (CRB) was derived, leading to the development of an efficient bisection-based algorithm known as RV-ML.
In contrast, this work extends the analysis to both overdetermined and underdetermined cases, allowing the mean of the random measurement matrix to be rank-deficient. A remarkable contribution is the proof that the equivalent MLE problem is convex and satisfies strong duality, strengthening previous quasiconvexity results. Moreover, it is shown that in underdetermined scenarios, the randomness in the measurement matrix can be beneficial for estimation under certain conditions.
In addition, a fast and unified implementation of the MLE solution, referred to as generalized RV-ML (GRV-ML), is proposed, which handles a more general case including both underdetermined and overdetermined systems. Extensive numerical simulations are provided to validate the theoretical findings.

\end{abstract}


% \begin{highlights}
% \end{highlights}

%% Keywords

\end{frontmatter}


\section{Introduction}
Linear regression model corresponding to the least squares (LS) solution has been widely applied in signal processing fields \cite{keyEst, raolin}. In practice, the measurement matrix in the linear regression model may have some uncertainty, and the LS solution ignoring the uncertainty leads to poor estimation performance. 
% \blue{The uncertainty arises naturally in settings such as calibration, where 
% $G$ is estimated from multiple trials and its deviation from the sample mean is Gaussian. In a quantized system, step-size $\Delta$ induces uniform noise of variance $\frac{\Delta^2}{12}$ \cite{CSboyd}. In the case of imperfect channel knowledge in a MIMO communication system, the error between true and estimated channel matrices is Gaussian distributed \cite{MIMOICSI}.}
Several methods have been proposed to model the uncertainty. The first is the errors-in-variables (EIV) model, which models the measurement matrix as an unknown deterministic matrix, and the noisy observation of the measurement matrix is available  
\cite{errormodel}. The joint maximum likelihood estimation of the  deterministic unknown matrix and the signal leads to the total least squares (TLS) solution \cite{TLS1, TLS}.
% Performing the joint estimation of the measurement matrix and the unknown deterministic parameter amounts to the total LS (TLS) solution \cite{TLS}.
The second is to model the measurement matrix belong to an elipsoidal set, and robust estimation is conducted by using the S procedure \cite{robusteldar, robustBeck}. The last is random variable (RV) model which treats the measurement matrix as the random matrix whose mean and variance are known \cite{SPLeldar, EldarCRB}. By marginalizing over the random matrix to obtain the loglikelihood, the maximum likelihood estimator (MLE) is developed \cite{SPLeldar}. 

In communication systems over multiple-input–multiple-output (MIMO) channels, the measurement matrix often represents a time-varying and inherently random channel \cite{EldarCRB}. In such cases, the RV model is especially suitable, as it characterizes the measurement matrix as a random perturbation around a nominal matrix, reflecting the statistical uncertainty in the channel state. This uncertainty is attributable to the imperfect channel state information.
Later, the estimation and detection for the RV model under low resolution quantization have also been extensively studied in depth  \cite{zhump, ligang, Sani, ningzhang, multifading}.

In the RV model in which the mean of the random measurement matrix has full column rank, \cite{SPLeldar} demonstrated that the MLE can be formulated as a quasiconvex optimization problem, which can be efficiently solved using bisection methods known as RV-ML. The computational complexity of this approach is primarily determined by the number of iterations in the bisection method, along with the complexity of solving the trust region problem. \cite{EldarCRB} proposed a sophisticated expectation maximization (EM) algorithm, but its computational complexity is also high, and convergence to the global optimum is not guaranteed. In this work, general cases including overdetermined and underdetermined where the rank of the mean of the random measurement matrix either has full column rank or rank deficient are studied.  It is shown that the MLE can be transformed into a convex optimization problem using the lifting technique by introducing an additional variable, where strong duality holds. It is shown that the global solution can be found, and an algorithm, referred to as generalized RV-ML (GRV-ML), with lower computation complexity compared to the previous methods is developed to address the general case. Finally, numerical simulations are conducted to illustrate the correctness of the established theoretical results.
The GRV-ML software implementing the proposed algorithm can be found in \cite{code}.

\section{Problem Formulation}
This section describes the problem setup and reviews the results presented in \cite{SPLeldar}. The linear regression with Gaussian model uncertainty is \cite{SPLeldar, EldarCRB}
\begin{align}\label{measmodel}
    \mathbf{y} =({\mathbf H}+{\mathbf E}){\mathbf x}+ \boldsymbol{\epsilon},
\end{align}
where $\mathbf{y}\in{\mathbb R}^{M}$ denotes the noisy measurements, ${\mathbf G}={\mathbf H}+{\mathbf E}$ is the random measurements matrix with mean ${\mathbf H}\in{\mathbb R}^{M\times N}$ and random component $\mathbf E$, where the $(m,n)$th element  $E_{mn}$ is independent and identically distributed (i.i.d.) and satisfying $E_{mn} \sim \mathcal{N}(0, \sigma^2_e)$, $\boldsymbol{\epsilon}\in{\mathbb R}^{M}$ denotes the additive white Gaussian noise (AWGN) independent of $\mathbf E$ and satisfying $\boldsymbol{\epsilon}\sim {\mathcal N}({\mathbf 0},\sigma_{{\epsilon}}^2{\mathbf I}_M)$. The nuisance parameters $\sigma_{{\epsilon}}^2$  and $\sigma_e^2$ are supposed to be known. 
Here, an example is presented to show that $\sigma_e^2$ and $\sigma_{{\epsilon}}^2$ can be acquired in some settings. In a quantized system where both the measurement matrix $\mathbf G$ and the measurements $\mathbf y$ are quantized where the stepsize of the quantizer is $\Delta$ \cite{, CSboyd}, then the quantizer noise follows ${\mathcal U}\left(-\Delta/2,\Delta/2\right)$ where ${\mathcal U}\left(a,b\right)$ denotes the uniform distribution in the interval $(a,b)$. In this setting, the variance of the quantization noise is $\frac{\Delta^2}{12}$ \cite{IRE}. If we take the quantization into consideration and approximate the quantization noise as Gaussian distribution, one has $\sigma_e^2 =\frac{\Delta^2}{12}$ and $\sigma_{{\epsilon}}^2 = \frac{\Delta^2}{12}+ \sigma_{\tilde{{\epsilon}}}^2$ with $\sigma_{\tilde{{\epsilon}}}^2$ being the variance of the measurement noise in the high resolution quantization setting.

Define the equivalent noise $\boldsymbol{\epsilon}_{\rm{eq}}$ as \cite{SPLeldar}
\begin{equation}\label{eqnoise}
\boldsymbol{\epsilon}_{\rm{eq}}={\mathbf E}{\mathbf x}+\boldsymbol{\epsilon},
\end{equation}
model (\ref{measmodel}) reduces to 
\begin{align}\label{measmodelred}
    \mathbf{y} ={\mathbf H}{\mathbf x}+\boldsymbol{\epsilon}_{\rm{eq}},
\end{align}
where $\boldsymbol{\epsilon}_{\rm{eq}}\sim {\mathcal N}({\mathbf 0},(\sigma_e^2\|\mathbf x\|_2^2+\sigma_{{\epsilon}}^2){\mathbf I}_M)$ \cite{zhuMLE}. Consequently, the MLE is formulated as 
\begin{align}\label{mle}
\widehat{\mathbf x}_{\rm ML}=\underset{\mathbf x}{\operatorname{argmin}}~ f({\mathbf y};\mathbf x),
\end{align}
where $f({\mathbf y};\mathbf x)$ denotes the negative loglikelihood function
\begin{align}\label{fxydef}
    f({\mathbf y};\mathbf x)=\frac{\|{\mathbf y}-{\mathbf H}{\mathbf x}\|_2^2}{2\left(\sigma_e^2\|\mathbf x\|_2^2+\sigma_{{\epsilon}}^2\right)}+\frac{M}{2}\log\left(\sigma_e^2\|\mathbf x\|_2^2+\sigma_{{\epsilon}}^2\right)
\end{align}
dropping some irrelevant terms. \cite{SPLeldar, EldarCRB} investigated the overdetermined case $M\geq N$ with $\mathbf H$ being full column rank, i.e., ${\rm rank}({\mathbf H})=N$, and showed that (\ref{mle}) is a quasiconvex optimization problem and bisection methods can be adopted to obtain $\widehat{\mathbf x}_{\rm ML}$ efficiently. In the following, we strengthen these results by showing the convexity of the equivalent MLE for both the overdetermined and underdetermined cases. 
% Without generality, we only consider the case $M\geq N$. When $M<N$, obviously $\hat{\mathbf x}_{\rm ML}=\mathbf{H}^\dagger\mathbf{y}$, where $\dagger$ denotes pseudo-inverse, achieving minimal $l_2$ norm and satisfying $\mathbf{H}\mathbf{x}={\mathbf y}$ is a solution of (\ref{mle}).

\section{A Convexity Result and A more Efficient Algorithm}
In the following, we show the original MLE (\ref{mle}) can be equivalently transformed to a convex optimization problem, and strong duality holds for arbitrary $M$ and $N$. Consequently, the Karush-Kuhn-Tucker (KKT) conditions are sufficient and necessary conditions for optimality \cite{cvxbook}. The optimal solutions are obtained via solving the KKT conditions, which can be solved efficiently.
% in two cases, of which one is to 
% find a root of a nonlinear equation with a scalar variable using a bisection method, and the other one exists infinite number of the estimators achieveing the minimum of the negative loglikelihood.

The singular value decomposition (SVD) is applied to $ \mathbf{H} $, i.e.,
\begin{align}\label{HSVD}
\mathbf{H} = \mathbf{U} \boldsymbol{\Sigma} \mathbf{V}^{\rm T},
\end{align} 
where $\mathbf{U} \in \mathbb{R}^{M \times M}$ and $\mathbf{V} \in \mathbb{R}^{N \times N}$ are orthogonal matrices, and $ \boldsymbol{\Sigma} \in \mathbb{R}^{M \times N} $ is a diagonal matrix containing the singular values $ \lambda_1, \lambda_2, \dots, \lambda_R $ on its diagonal, with $ R={\rm rank}({\mathbf H}) \leq \operatorname{min}(M,N) $. The singular values are arranged in descending order: $ \lambda_1 \geq \lambda_2 \geq \dots \geq \lambda_R > 0 $. Note that the orthogonal transform preserves the $l_2$ norm of the original vector, i.e., $\|{\mathbf U}{\mathbf y}\|_2=\|{\mathbf y}\|_2$ and $\|{\mathbf V}{\mathbf x}\|_2=\|{\mathbf x}\|_2$. Define two vectors $\widetilde{\mathbf y}$ and $\widetilde{\mathbf x}$ as
\begin{subequations}\label{defxy}
\begin{align}
\widetilde{\mathbf{y}}={{\mathbf{U}}^{\text{T}}}\mathbf{y},\label{defineytilde}\\
\widetilde{\mathbf{x}}={\mathbf V}^{\rm T}{\mathbf x}.\label{definextilde}
\end{align} 
\end{subequations}
Note that there exists a one to one mapping between $\widetilde{\mathbf{x}}$ and ${\mathbf x}$. 
Let $a:b=\{a,a+1,\cdots,b\}$ and $\mathbf y_{a:b} = [y_a, y_{a+1}, \dots, y_b]^{\rm T}\in{\mathbb R}^{b-a+1} $ denotes the part extracting from $\mathbf y$.
For $a>b$, $a:b$ is an empty set, $\mathbf y_{a:b}$ does not exist and $\|\mathbf y_{a:b}\|_2^2=0$.
Substituting (\ref{HSVD}) and (\ref{defxy}) in (\ref{fxydef}), the original MLE (\ref{mle}) can be equivalently transformed to another problem, i.e., by solving $\widetilde{\mathbf{x}}$ defined in (\ref{definextilde}) via
% \begin{small}
\begin{align}\label{MLtildefunc}
\widehat{\widetilde{\mathbf{x}}}_{\mathrm{ML}}=\underset{\widetilde{\mathbf x}}{\operatorname{argmin}}\frac{\|{{\widetilde{\mathbf{y}}}_{1:R}}- {\boldsymbol\lambda }\odot {{{\mathbf{\widetilde{x}}_{1: R}}}}\|_{2}^{2}+\|{{\widetilde{\mathbf{y}}}_{R+1:M}}\|_{2}^{2}}{2(\sigma _{e}^{2}\|{{\widetilde{\mathbf{x}}}}\|_{2}^{2}+\sigma_{{\epsilon}}^{2})}+\frac{M}{2}\log \left( \sigma _{e}^{2}\|{{\widetilde{\mathbf{x}}}}\|_{2}^{2}+\sigma_{{\epsilon}}^{2} \right),
\end{align}
% \end{small}
where 
% $\widetilde{\mathbf{y}}_{1:R} = [\widetilde{y}_1,\widetilde{y}_2,\cdots,\widetilde{y}_R]^{\rm T}$, $\widetilde{\mathbf{y}}_{R+1:M} = [\widetilde{y}_{R+1},\widetilde{y}_{R+2},\cdots,\widetilde{y}_M]^{\rm T}$, 
$\boldsymbol\lambda=[\lambda_1 , \lambda_2, \dots ,\lambda_R]^{\rm T}$
% , $\mathbf{\widetilde{x}}_{1\sim R} = [\widetilde{x}_1,\widetilde{x}_2,\cdots,\widetilde{x}_R]^{\rm T}$
and $\odot$ represents Hadamard product. It can be seen that (\ref{MLtildefunc}) also applies when $R=M$.

To further delve into the structure of (\ref{MLtildefunc}), we perform reparameterization by introducing two slack variables, i.e.,
\begin{subequations}\label{defut}
\begin{align}
\mathbf u=\frac{\widetilde{\mathbf x}}{\sqrt{\sigma _{e}^{2}\|{{\widetilde{\mathbf{x}}}}\|_{2}^{2}+\sigma_{{\epsilon}}^{2}}},\\
t=\frac{1}{\sqrt{\sigma _{e}^{2}\|{{\widetilde{\mathbf{x}}}}\|_{2}^{2}+\sigma_{{\epsilon}}^{2}}}.\label{tdef}
\end{align}    
\end{subequations}
Obviously, we obtain another equivalent but constrained optimization problem 
\begin{subequations}\label{solvehtu}
\begin{align}
     &\underset{{\mathbf u},t>0}{\operatorname{minimize}}
&&h({\mathbf u},t) \label{optprobobj}\\
     &{\rm subject ~to}&&\sum_{j=1}^{R} \sigma _{e}^{2} u_{j}^{2}+ \sum_{i=R+1}^{N} \sigma _{e}^{2} u_{i}^{2} +\sigma_{{\epsilon}}^2 t^{2}=1,\label{optprobcons}
\end{align} 
\end{subequations}
where $h({\mathbf u},t)$ is 
\begin{align}\label{hutdef}
h({\mathbf u},t)&\triangleq \frac{1}{2}\sum_{j=1}^{R}\left( \widetilde{y}_{j}t-\lambda_{j} u_{j}\right)^{2}+\frac{1}{2}\|\widetilde{\mathbf y}_{R+1:M}\|_2^2t^2-M \log t \notag\\
&= \sum_{j=1}^{R} \frac{1}{2}\left(  \widetilde{{y}}_{j}^2t^2 + \lambda_j^2 u_j^2 - 2t \widetilde{{y}}_{j} \lambda_j u_j \right) + \frac{1}{2}\|\widetilde{\mathbf y}_{R+1:M}\|_2^2t^2  - M \log t.
\end{align}
It can be seen that the number of variables is $N+1$.
Upon examining the constraint (\ref{optprobcons}), we observe that it is symmetric with respect to $\mathbf u$, i.e., for any pair $({\mathbf u}, t)$ satisfying (\ref{optprobcons}), the pair $(-{\mathbf u}, t)$ also satisfies (\ref{optprobcons}) when $t > 0$. Leveraging the symmetric property of the constraint (\ref{optprobcons}), the objective function $h({\mathbf u}, t)$ attains its minimum when the signs of $u_j$ and $\widetilde{y}_{j}$ align for $j = 1,2,\cdots, R$. Consequently, we obtain another equivalent problem with an additional nonnegative constraint ${\mathbf o}\geq {\mathbf 0}$, where $\geq$ represents componentwise inequality,  formulated as 
\begin{subequations}\label{solvehsymtu}
\begin{align}
     &\underset{{\mathbf o}\geq {\mathbf 0},t>0}{\operatorname{minimize}}
&&g({\mathbf o},t) \label{optprobsymobj}\\
     &{\rm subject ~to}&&\sum_{j=1}^{R} \sigma _{e}^{2} o_{j}^{2}+\sum_{i=R+1}^{N} \sigma _{e}^{2} o_{i}^{2} +\sigma_{{\epsilon}}^2 t^{2}=1,\label{optprobsymcons}
\end{align} 
\end{subequations}
where $g({\mathbf o},t)$ is defined as
% \begin{figure*}
%     \begin{small}
\begin{align}\label{got}
g({\mathbf o},t)=\sum_{j=1}^{R} \frac{1}{2}\left(  \widetilde{{y}}_{j}^2t^2 + \lambda_j^2 o_j^2 - 2t|\widetilde{{y}}_{j}|\lambda_j o_j \right) + \frac{1}{2}\|\widetilde{\mathbf y}_{R+1:M}\|_2^2t^2  - M \log t.
\end{align}
%     \end{small}
% \end{figure*}
Obviously, the optimal ${\mathbf o}^{\star}$ and ${\mathbf u}^{\star}$ are related via 
\begin{align}\label{ustar}
{\mathbf u}_{1:R}^{\star}={\mathbf o}_{1:R}^{\star}\odot{\rm sign}\left(\widetilde{\mathbf y}_{1:R}\right).
\end{align}
And ${\mathbf o}_{R+1:N}^{\star}$ only has to satisfy the constraint (\ref{optprobsymcons}) and ${\mathbf u}_{R+1:N}^{\star}={\mathbf o}_{R+1:N}^{\star}$.

For problem (\ref{solvehsymtu}), define \begin{align}\label{defwz}
{\mathbf w}={\mathbf o}^2,z=t^2,
\end{align}
where ${\mathbf o}^2$ denotes the componentwise squared operation. We perform the re-parametrization trick again to transform problem (\ref{solvehsymtu}) into another equivalent optimization problem, i.e.,
% \begin{figure*}[htb!]
% \begin{small}
\begin{subequations}\label{optfunczw}
\begin{align}
    &\underset{{\mathbf w},z>0}{\operatorname{minimize}}~&&
    \sum_{j=1}^{R} \frac{1}{2} \left( \widetilde{{y}}_{j}^2z + \lambda_j^2 w_{j} - 2\left|\widetilde{y}_{j}\right|\lambda_j \sqrt{w_{j} z} \right) + \frac{1}{2} \|\widetilde{{\mathbf y}}_{R+1:M}\|_2^2 z-\frac{1}{2} M \log z,\label{optobj} \\
    &{\rm subject ~to}&&\sum_{j=1}^{R} \sigma _{e}^{2} w_j + \sum_{i=R+1}^{N}\sigma _{e}^{2} w_i+\sigma_{{\epsilon}}^2 z=1,\label{equcons}\\
    & &&w_j \geq 0,\  j = 1,2,\cdots ,R,\label{inequconsj}\\
    & &&w_i \geq 0,\ i = R+1,R+2,\cdots,N.\label{inequconsi}
\end{align} 
\end{subequations}

In the following, we will prove that (\ref{optfunczw}) is a convex optimization problem. In addition, the constraints (\ref{equcons})-(\ref{inequconsi}) are all linear equalities and inequalities.
According to Slater's constraint qualification, strong duality holds \cite{cvxbook}, thus KKT conditions are sufficient and necessary conditions for optimality. We present a numerical efficient method to solve the KKT equation. Once (\ref{optfunczw}) is globally solved, the original MLE can be obtained accordingly. 

For problem (\ref{optfunczw}), the terms $  \widetilde{{y}}_{j}^2z$, $ \lambda_j^2 w_{j}$ and $ \|\widetilde{{\mathbf y}}_{R+1:M}\|_2^2 z$ are linear functions with respect to variable $z$ and $w_j$, the component $-\frac{1}{2} M \log z$  is a convex function. In the following, we present the convexity of $q({\mathbf w},z)=-\sqrt{z}\sum\limits_{j=1}^R C_j\sqrt{w_j}$ in the following proposition for $C_j>0$.
\begin{proposition}
For $C_j>0,j=1,2,\cdots,R$, $q({\mathbf w},z)=-\sqrt{z}\sum\limits_{j=1}^R C_j\sqrt{w_j}$ is a convex function with respect to $[{\mathbf w}^{\rm T},z]^{\rm T}$.
\end{proposition}
\begin{proof}
According to the structure of $q({\mathbf w},z)$, its Hessian matrix $\nabla^2q({\mathbf w},z)$ has the following form
\begin{align}\label{henssian}
    \nabla^2q({\mathbf w},z)=\left[\begin{array}{cc}
        \frac{\partial^2q({\mathbf w},z)}{\partial z^2} & {\mathbf d}^{\rm T} \\
        {\mathbf d} & {\rm diag}\left(\mathbf p\right)
    \end{array}\right],
\end{align}
where $\frac{\partial^2q({\mathbf w},z)}{\partial z^2}=\frac{1}{4}\sum_{j=1}^{R}   {C_j w_j^{\frac{1}{2}}z^{-\frac{3}{2}}}$, $p_j=\frac{\partial^2 q({\mathbf w},z)}{\partial w_j^2} =\frac{1}{4}  {C_j w_j^{-\frac{3}{2}}z^{\frac{1}{2}}}$, $d_j=\frac{\partial^2 q({\mathbf w},z)}{\partial z \partial w_j} = -\frac{1}{4}C_{j}w_{j}^{-\frac{1}{2}}z^{-\frac{1}{2}}$.
According to Schur's lemma \cite{cvxbook}, $\nabla^2q({\mathbf w},z)\succeq {\mathbf 0}$ if and only if 
\begin{subequations}
\begin{align}
{\rm diag}\left(\mathbf p\right)\succeq {\mathbf 0},\frac{\partial^2q({\mathbf w},z)}{\partial z^2}\geq 0;\label{part1}\\
\frac{\partial^2q({\mathbf w},z)}{\partial z^2}-{\mathbf d}^{\rm T}{\rm diag}^{-1}\left(\mathbf p\right){\mathbf d}\geq 0.\label{part2}
\end{align}
\end{subequations}
(\ref{part1}) is satisfied obviously, and (\ref{part2}) is satisfied with equality, establishing the convexity of  $q({\mathbf w},z)$.
\end{proof}

According to convexity preservation property such that the nonnegative weighted sum of convex functions is also convex \cite{cvxbook}, the objective function of (\ref{optfunczw}) is convex. According to Slater's constraint qualification \cite{cvxbook}, the constraints, comprising a linear equality and linear inequalities, satisfy the regularity conditions under which strong duality holds and the KKT conditions provide both necessary and sufficient conditions for optimality. Let $\nu$ be a dual variable associated with the equality constraint (\ref{equcons}) and $\boldsymbol{\eta}\in \mathbb{R}^N$ be the dual variable associated with the inequality constraints (\ref{inequconsj}) and (\ref{inequconsi}). The Lagrangian function $ L({\mathbf w}, z, \nu,\boldsymbol{\eta}) $ is constructed as
\begin{align}\label{Larg}
    L({\mathbf w},z,\nu,\boldsymbol{\eta})=&\sum_{j=1}^{R} \frac{1}{2} \left( \widetilde{y}_{1j}^2 z + \lambda_j^2 w_{j} - 2\left|\widetilde{y}_{j}\right|\lambda_j \sqrt{w_{j} z} \right) +  \frac{1}{2} \|\widetilde{{\mathbf y}}_{R+1:M}\|_2^2z \notag\\&-\frac{1}{2} M \log z+
    \nu \left(\sum_{j=1}^{R} \sigma _{e}^{2} w_j + \sum_{i=R+1}^{N}\sigma _{e}^{2} w_i+\sigma_{{\epsilon}}^2 z-1\right) \notag\\&+ \sum_{j=1}^{R}(- \eta_j w_j)+ \sum_{i=R+1}^{N}(- \eta_i w_i).
\end{align}
The KKT conditions are 
\begin{subequations}{\label{KKTcon}}
    \begin{align}
    &\nabla_{w_j}L({\mathbf w},z,\nu,\boldsymbol{\eta})=\frac{1}{2}\lambda_j^2-\frac{1}{2}\left|\widetilde{y}_{j}\right|\lambda_j z^{\frac{1}{2}}w_j^{-\frac{1}{2}}+\nu\sigma_e^2 - \eta_j=0,\ j = 1,2,\cdots ,R,\label{KKTwj}\\
    &\nabla_{w_i}L({\mathbf w},z,\nu,\boldsymbol{\eta})=\nu\sigma_e^2 - \eta_i=0,\ i = R+1,R+2,\cdots,N,\label{KKTwi}\\
    &\nabla_{z}L({\mathbf w},z,\nu,\boldsymbol{\eta})=  \sum_{j=1}^{R} \frac{1}{2}\widetilde{y}_{j}^2 -  \sum_{j=1}^{R}\frac{1}{2}\left|\widetilde{y}_{j}\right|\lambda_j w_j^{\frac{1}{2}}z^{-\frac{1}{2}}  + \frac{1}{2} \|\widetilde{{\mathbf y}}_{R+1:M}\|_2^2-\frac{1}{2} M  z^{-{1}}\notag\\&\ \ \ \ \ \ \ \ \ \ \ \ \ \ \ \ \ \ \ \ \ \ \ +\nu\sigma_{{\epsilon}}^2
    ={0},\label{KKTz}\\
    &\sum_{j=1}^{R} \sigma _{e}^{2} w_j + \sum_{i=R+1}^{N}\sigma _{e}^{2} w_i+\sigma_{{\epsilon}}^2 z=1,\label{sum1}\\
    &w_j\geq 0,\ \eta_j\geq 0,\ j = 1,2,\cdots ,R,\label{dualvarj}\\
    &w_i\geq 0,\ \eta_i\geq 0,\ i = R+1,R+2,\cdots,N,\label{dualvari}\\
    &\eta_j w_j =0,\ j = 1,2,\cdots ,R,\label{KKTetaj}\\
    &\eta_i w_i =0,\ i = R+1,R+2,\cdots,N,\label{KKTetai}
    \end{align}
\end{subequations}
where (\ref{KKTwj})-(\ref{KKTz}) are the zero gradient conditions, (\ref{sum1}) is the primal constraint, (\ref{dualvarj}) and (\ref{dualvari}) are the dual constraints and (\ref{KKTetaj}) and (\ref{KKTetai}) are the complementary slackness conditions.

Let us solve these equations to find $\mathbf{w}^{\star}$, $z^{\star}$, $\nu^{\star}$ and $\boldsymbol{\eta}^{\star}$. According to (\ref{KKTwi}) and (\ref{KKTetai}), for $i = R+1,R+2,\cdots,N$, one has 
% $\eta_i$ have to satisfy the following set of equations for $i = R+1,R+2,\cdots,N$,
\begin{subequations}\label{etaiall}
\begin{align}
&\eta_i=\nu\sigma_e^2,\label{nueta}\\
&\eta_i w_i =0.\label{wieta}
\end{align}
\end{subequations}
(\ref{wieta}) implies that either $\eta_i=0$ or $w_i=0$. (\ref{nueta}) shows that $\eta_i=0$ implies $\nu=0$. Therefore, in the following, we discuss two cases: $\nu=0$ and $\nu\neq 0$.

If $\nu\neq0$, according to  (\ref{dualvari}) and (\ref{nueta}), we have $\eta_i> 0$. Thus, for $i = R+1,R+2,\cdots,N$, equation (\ref{wieta}) holds only if 
\begin{align}\label{wi0}
    w_i=0.
\end{align}
One can solve the equations (\ref{KKTwj})-(\ref{dualvarj}) and (\ref{KKTetaj}) to find $\boldsymbol{\eta}$, $w_j$, $z$ and $\nu$.

In the following subsections, we discuss the above two cases separately, i.e., $\nu =0$ and $\nu\neq 0$.
Note that when $R=N$, $\{w_i,i = R+1,R+2,\cdots,N\}$ is an empty set. Therefore, for a given $\nu$, we further distinguish the rank-deficient and full column rank cases, i.e., $R\leq N-1$ and $R=N$.

\subsection{Case $\nu = 0$}\label{subcasenu0}
\subsubsection{Case $R\leq N-1$}\label{subsubcase1}

In this case, $R\leq N-1$ means $\{w_i,i = R+1,R+2,\cdots,N\}$ is not an empty set, i.e., the primal variables $\mathbf{w}_{R+1:N}$ and the dual variables $\boldsymbol{\eta}_{R+1:N}$ exist.

When $\nu=0$, (\ref{nueta}) shows that
\begin{align}\label{etai0}
    \eta_i=0,\ i = R+1,R+2,\cdots,N,
\end{align}
and (\ref{KKTetai}) is always satisfied. Substituting $\nu=0$ in (\ref{KKTwj}) and (\ref{KKTz}), one obtains
\begin{subequations}\label{KKTwithnu0}
    \begin{align}
        &\frac{1}{2}\lambda_j^2-\frac{1}{2}\left|\widetilde{y}_{j}\right|\lambda_j z^{\frac{1}{2}}w_j^{-\frac{1}{2}}- \eta_j=0,\ j = 1,2,\cdots ,R,\label{KKTwnu0}\\
        &\sum_{j=1}^{R} \frac{1}{2}\widetilde{y}_{j}^2 -  \sum_{j=1}^{R}\frac{1}{2}\left|\widetilde{y}_{j}\right|\lambda_j w_j^{\frac{1}{2}}z^{-\frac{1}{2}}  + \frac{1}{2} \|\widetilde{{\mathbf y}}_{R+1:M}\|_2^2-\frac{1}{2} M  z^{-{1}}={0}.\label{KKTznu0}
    \end{align}
\end{subequations}
According to (\ref{KKTwnu0}), for $j = 1,2,\cdots, R$, $w_j$ can be expressed in terms of $z$ and $\eta_j$ as
\begin{align}\label{wjresult}
    w_j = \frac{\widetilde{y}_{j}^2 \lambda_j^2 z}{(\lambda_j^2 - 2\eta_j)^2}.
\end{align}
Substituting (\ref{wjresult}) into (\ref{KKTznu0}) and eliminating variables $\mathbf w_{1:R}$, one obtains 
\begin{align}\label{zequ}
    \sum_{j=1}^{R} \frac{1}{2}\widetilde{y}_{j}^2 -  \sum_{j=1}^{R}\frac{\widetilde{y}_{j}^2 \lambda_j^2}{2|\lambda_j^2 - 2\eta_j|}  + \frac{1}{2} \|\widetilde{{\mathbf y}}_{R+1:M}\|_2^2 -\frac{1}{2} M  z^{-{1}}
    ={0}.
\end{align}
According to (\ref{zequ}), $z$ can be solved as
\begin{align}\label{zresult}
    z=\frac{M}{\sum_{j=1}^{R}\widetilde{y}_{j}^2 \left(1-\frac{\lambda_j^2}{|\lambda_j^2 - 2\eta_j|} \right) +\|\widetilde{{\mathbf y}}_{R+1:M}\|_2^2}
\end{align}
in term of the dual variables $\boldsymbol{\eta}_{1:R}$.
According to (\ref{wjresult}) and (\ref{zresult}), $\mathbf w_{1:R}$ can be represented by dual variables $\boldsymbol{\eta}_{1:R}$ as
\begin{align}\label{wjresultnoz}
    w_j = \frac{\widetilde{y}_{j}^2 \lambda_j^2 }{(\lambda_j^2 - 2\eta_j)^2}\frac{M}{\left(\sum_{j=1}^{R}\widetilde{y}_{j}^2 \left(1-\frac{\lambda_j^2}{|\lambda_j^2 - 2\eta_j|} \right) +\|\widetilde{{\mathbf y}}_{R+1:M}\|_2^2\right)},\ j = 1,2,\cdots, R,
\end{align}
by substituting (\ref{zresult}) into (\ref{wjresult}).

Note that $\mathbf w_{1:R}$ (\ref{wjresultnoz}) must satisfy complementary slackness condition (\ref{KKTetaj}). According to the definition of $\widetilde{\mathbf{y}}$ given as $\widetilde{\mathbf{y}}=\mathbf{U}^{\rm T}\mathbf{y}=\boldsymbol{\Sigma}\widetilde{\mathbf{x}}+\mathbf{U}^{\rm T}\boldsymbol{\epsilon}_{\rm eq}$, $\widetilde{\mathbf{y}}$ follows Gaussian distribution $\mathcal{N}(\boldsymbol{\Sigma}\widetilde{\mathbf{x}},(\sigma_e^2\|\mathbf x\|_2^2+\sigma_{{\epsilon}}^2){\mathbf I}_M)$, which implies that the probability of $\widetilde{\mathbf{y}}=\mathbf{0}$ is zero. Without loss of generality, we let $\widetilde{\mathbf{y}}\neq\mathbf{0}$. In this way, $\mathbf{w}_{1:R}\neq\mathbf{0}$ according to (\ref{wjresultnoz}). Utilizing the complementary slackness condition (\ref{KKTetaj}), one has
\begin{align}\label{etaj0}
    \eta_j=0,\ j = 1,2,\cdots, R.
\end{align}

Substituting (\ref{etaj0}) into  (\ref{zresult}) and (\ref{wjresultnoz}), $z$ and $\mathbf{w}_{1:R}$ are
\begin{subequations}\label{zwresultend}
\begin{align}
    &z=\frac{M}{\|\widetilde{{\mathbf y}}_{R+1:M}\|_2^2},\label{zresultend}\\
    &w_j = \frac{\widetilde{y}_{j}^2  }{\lambda_j^2}\frac{M}{\|\widetilde{{\mathbf y}}_{R+1:M}\|_2^2},\ j = 1,2,\cdots, R,\label{wjresultend}
\end{align}
\end{subequations}
respectively.

Until now, the dual variables $\nu$ and $\boldsymbol{\eta}$, the primal variables $z$ and $\mathbf{w}_{1:R}$ are determined, and only the primal variables $\mathbf{w}_{R+1:N}$ remain unknown. 
Note that $\mathbf{w}_{R+1:N}$ must satisfy the constraint (\ref{sum1}).
Substituting $z$ (\ref{zresultend}) and $\mathbf{w}_{1:R}$ (\ref{wjresultend}) into  (\ref{sum1}), one obtains
\begin{align}\label{sumnozw}
    \left(\sigma_e^2 \sum_{j=1}^{R} \frac{ \widetilde{y}_{j}^2}{\lambda_j^2} + \sigma_{{\epsilon}}^2 \right) \frac{M}{\|\widetilde{{\mathbf y}}_{R+1:M}\|_2^2} + \sum_{i=R+1}^{N} \sigma_e^2 w_i =1.
\end{align}

Because $\mathbf{w}_{R+1:N}$ must satisfy $\mathbf{w}_{R+1:N}\geq \mathbf 0$ (\ref{dualvarj}), (\ref{sumnozw}) demonstrates that the optimal solution $\mathbf{w}_{R+1:N}$ exists if and only if 
\begin{align}\label{sumgeq0}
    S \triangleq \frac{1}{\sigma_e^2}-\left(\sum_{j=1}^{R} \frac{ \widetilde{y}_{j}^2}{\lambda_j^2} + \frac{\sigma_{{\epsilon}}^2}{\sigma_e^2 } \right) \frac{M}{\|\widetilde{{\mathbf y}}_{R+1:M}\|_2^2}\geq 0.
\end{align}
If (\ref{sumgeq0}) is not satisfied, this indicates that the precondition $\nu=0$ does not hold. In this setting, $\nu\neq 0$, and the details of solving the KKT conditions (\ref{KKTcon}) are deferred in Subsection \ref{subcasenunot0}.
If (\ref{sumgeq0}) is satisfied, for $R=N-1$, a unique solution $w_N$ exists. For $R<N-1$, any $\mathbf{w}_{R+1:N}$ satisfying (\ref{sumnozw}) and $\mathbf{w}_{R+1:N}\geq \mathbf{0}$ simultaneously are the optimal solutions.

Finally, the optimal $\widetilde{\mathbf x}^{\star}$ can be recovered from the optimal $\mathbf w^\star$ and $z^\star$. According to (\ref{defwz}) and $\mathbf{o}_{1:R}\geq \mathbf{0}$, one has $\mathbf{o}^\star_{1:R}=\sqrt{\mathbf{w}^\star_{1:R}}$ and 
\begin{align}\label{tstar}
    t^\star = \sqrt{z^\star}.
\end{align}
According to (\ref{ustar}), $\mathbf{u}^\star$ is
\begin{align}\label{utend}
&\mathbf{u}^\star_{1:R}=\sqrt{\mathbf{w}^\star_{1:R}}\odot \operatorname{sign}(\widetilde{\mathbf{y}}_{1:R}),\notag\\
&\mathbf{u}^\star_{R+1:N}=\pm \sqrt{\mathbf{w}^\star_{R+1:N}},
\end{align}
i.e., the sign of $\mathbf{u}^\star_{1:R}$ is consistent with that of $\widetilde{\mathbf{y}}_{1:R}$, and the sign of $\mathbf{u}^\star_{R+1:N}$ is arbitrary.

According to (\ref{defut}), $\widetilde{\mathbf x}^\star=\mathbf{u}^\star/t^\star$.
Utilizing (\ref{tstar}) and (\ref{utend}), one has
\begin{align}\label{xstar}
    &\widetilde{\mathbf x}^\star_{1:R} = \sqrt{\mathbf{w}^\star_{1:R}/z^\star}\odot \operatorname{sign}(\widetilde{\mathbf{y}}_{1:R}),\notag\\
    &\widetilde{\mathbf x}^\star_{R+1:N}=\pm \sqrt{\mathbf w_{R+1:N}^\star/  z^\star}.
\end{align}
Substituting (\ref{zwresultend}) into (\ref{xstar}), $\widetilde{\mathbf x}^\star_{1:R}$ can be further simplified as
\begin{align}\label{widetildexjMLnu0}
\widetilde{\mathbf x}^\star_{1:R}=\widetilde{\mathbf y}_{1:R}\oslash \boldsymbol{\lambda},
\end{align}
where $\oslash$ denote componentwise division. 

% and the result in (\ref{zwresultend}) and (\ref{utend}), the optimal $\widetilde{\mathbf x}^\star$ is 
% \begin{align}\label{widetildexjMLnu0}
% &\widetilde{\mathbf x}^\star_{1:R}=\widetilde{\mathbf y}^\star_{1:R}\oslash \boldsymbol{\lambda},\notag\\
% &\widetilde{\mathbf x}^\star_{R+1:N}=\pm \sqrt{\mathbf w_{R+1,N}^\star/  z^\star},
% \end{align}
% where $\oslash$ denote componentwise division. 
% And the optimal $\widetilde{x}_i^\star$ for $i = R+1,R+2,\cdots,N$ is $\widetilde{x}_i^\star=\pm \sqrt{w_i^\star/z^\star}$, whose sign is uncertain. 

\subsubsection{Case $R=N$}\label{subsubcase2}
In this case, the set $\{w_i \mid i = R+1, R+2, \dots, N\}$ is empty, i.e., the primal variables $\mathbf{w}_{R+1:N}$ and the dual variables $\boldsymbol{\eta}_{R+1:N}$ vanish. Consequently, the zero-gradient condition (\ref{KKTwi}), the dual constraint (\ref{dualvari}), and the complementary slackness condition (\ref{KKTetai}) are no longer applicable and are therefore removed. The primal constraint (\ref{sum1}) is accordingly simplified to be
\begin{align}\label{sum1nowi}
    \sum_{j=1}^{N} \sigma _{e}^{2} w_j + \sigma_{{\epsilon}}^2 z=1.
\end{align}

Despite the removed equations, the solutions for the primal variables $\mathbf{w}$ and $z$ in (\ref{zwresultend}), as well as the dual variables $\boldsymbol{\eta}$ in (\ref{etaj0}), can still be derived by following the same procedure detailed in Subsubsection \ref{subsubcase1}, specifically from (\ref{KKTwithnu0}) to (\ref{zwresultend}).

Substituting (\ref{zwresultend}) into (\ref{xstar}), one has
\begin{align}\label{subcase2xstar}
\widetilde{\mathbf x}^\star=\widetilde{\mathbf y}\oslash \boldsymbol{\lambda}.
\end{align}

According to (\ref{definextilde}), $\widehat{\mathbf x}_{\rm ML}$ is 
\begin{align}\label{xMLend}
    \widehat{\mathbf x}_{\rm ML}={\mathbf V}{\widetilde{\mathbf x}^\star}.
\end{align}

For the optimization problem defined in (\ref{optfunczw}), an alternative interpretation can be considered when $\nu = 0$. In this case, constraint (\ref{equcons}) is eliminated, and (\ref{inequconsj}) is the implicit constraint of the objective function (\ref{optobj}). Since $\mathbf{w}_{R+1:N}$ no longer appears in the problem, it becomes decoupled from the optimization and constraint (\ref{inequconsi}) is automatically satisfied. As a result, problem (\ref{optfunczw}) reduces to the unconstrained optimization problem with objective function (\ref{optobj}) and variables $\mathbf{w}_{1:R}$ and $z$, requiring only the zero gradient conditions (\ref{KKTwj}) and (\ref{KKTz}) to be satisfied. These $R+1$ equations determine the optimal $\mathbf{w}_{1:R}^\star$ and $z^\star$, while simultaneously ensuring that constraints (\ref{inequconsj}) and (\ref{inequconsi}) are fulfilled. Therefore, when $\nu = 0$, the original constrained optimization problem is effectively transformed into an unconstrained one and the solution is (\ref{subcase2xstar}).

% Equation (\ref{sumwi}) implies that $w_i$ can be arbitrary satisfying primal constraint (\ref{sum1}) and dual constraint (\ref{dualvari}). In this case, the equality constrained optimization problem  (\ref{optfunczw})
% reduced to an equivalent unconstrained problem because the dual variable $\nu=0$ eliminates the equality constraint in Lagrangian function (\ref{Larg}). One can get the optimal $\mathbf{w}^\star$ and $z^\star$ using equation (\ref{zwresultend}), and get the ML result $\widehat{\mathbf x}_{\rm ML}$ using equation (\ref{widetildexjMLnu0}) and (\ref{xMLend}).

% solve the equations (\ref{KKTwj}), (\ref{KKTz}), (\ref{dualvarj}) and (\ref{KKTetaj}) to find $\eta_j$, $\mathbf{w}$ and $z$ while satisfying equation (\ref{sum1}).

\subsection{Case $\nu \neq 0$}\label{subcasenunot0}
\subsubsection{Case $R\leq N-1$}\label{subsubcase3}

In this case, $R\leq N-1$ means $\{w_i,i = R+1,R+2,\cdots,N\}$ is not an empty set, i.e., the primal variables $\mathbf{w}_{R+1:N}$ and the dual variables $\boldsymbol{\eta}_{R+1:N}$ exist.

When $\nu \neq 0$, (\ref{dualvari}) and (\ref{nueta}) show that 
\begin{align}\label{etainu}
    \eta_i=\nu\sigma_e^2> 0,\ i = R+1,R+2,\cdots,N,
\end{align}
and (\ref{wieta}) can only hold if
\begin{align}\label{wi0}
    w_i=0,\ i = R+1,R+2,\cdots,N.
\end{align}
% From equation (\ref{dualvari}) and (\ref{nueta}), we know $\eta_i> 0$ in this case. Thus, we have analyzed that for $i = R+1,R+2,\cdots,N$, equation (\ref{wieta}) can only holds if $w_i=0$ in equation (\ref{wi0}). 
% Variables $\boldsymbol{\eta}$, $w_j$, $z$ and $\nu$ remain unknown. 
Substituting (\ref{wi0}) into  (\ref{sum1}), one obtains
\begin{align}\label{sumnowi}
    \sum_{j=1}^{R} \sigma _{e}^{2} w_j +\sigma_{{\epsilon}}^2 z=1,
\end{align}
which involves with $z$ and $\mathbf{w}_{1:R}$. According to (\ref{KKTwj}), $\mathbf{w}_{1:R}$ can be solved as
\begin{align}\label{wjwithzwi0}
    w_j = \frac{\widetilde{y}_{j}^2 \lambda_j^2}{(\lambda_j^2 + 2\nu\sigma_e^2- 2\eta_j)^2}z,\ j = 1,2,\cdots, R
\end{align}
in terms of $z$, $\nu$ and $\boldsymbol{\eta}_{1:R}$.

Substituting (\ref{wjwithzwi0}) into (\ref{sumnowi}) and eliminating variables $\mathbf{w}_{1:R}$, one has
\begin{align}\label{sumz}
    \sum_{j=1}^{R} \sigma _{e}^{2} \frac{\widetilde{y}_{j}^2 \lambda_j^2}{(\lambda_j^2 + 2\nu\sigma_e^2- 2\eta_j)^2}z +\sigma_{{\epsilon}}^2 z=1.
\end{align}
According to (\ref{sumz}), $z$ can be represented in terms of dual variables $\nu$ and $\boldsymbol{\eta}_{1:R}$ as
\begin{align}\label{zstarwi0}
    z = \frac{1}{\sum_{j=1}^{R}\frac{ \sigma_e^2 \widetilde{y}_{j}^2 \lambda_j^2}{(\lambda_j^2 + 2\nu\sigma_e^2- 2\eta_j)^2} + \sigma_{{\epsilon}}^2 }.
\end{align}
According to (\ref{wjwithzwi0}) and (\ref{zstarwi0}), $\mathbf{w}_{1:R}$ can be represented by dual variables $\nu$ and $\boldsymbol{\eta}_{1:R}$ as
\begin{align}\label{wjresultwi0}
    w_j = \frac{1}{\left( \sum_{j=1}^{R} \frac{ \sigma_e^2 \widetilde{y}_{j}^2 \lambda_j^2}{(\lambda_j^2 + 2\nu\sigma_e^2- 2\eta_j)^2} + \sigma_{{\epsilon}}^2 \right)}\frac{\widetilde{y}_{j}^2 \lambda_j^2}{(\lambda_j^2 + 2\nu\sigma_e^2- 2\eta_j)^2},\ j = 1,2,\cdots, R,
\end{align}
by substituting (\ref{zstarwi0}) into (\ref{wjwithzwi0}).

Note that $\mathbf w_{1:R}$ (\ref{wjresultwi0}) must satisfy complementary slackness condition (\ref{KKTetaj}). Similar to the previous analysis shown in Subsubsection \ref{subsubcase2}, we let $\widetilde{\mathbf{y}}\neq\mathbf{0}$. In this way, $\mathbf{w}_{1:R}\neq\mathbf{0}$ according to (\ref{wjresultwi0}). Utilizing the complementary slackness condition (\ref{KKTetaj}), one has
\begin{align}\label{etaj0wi0}
    \eta_j=0,\ j = 1,2,\cdots, R.
\end{align}

Substituting (\ref{etaj0wi0}) into (\ref{zstarwi0}) and (\ref{wjresultwi0}), the primal variables $z$ and $\mathbf{w}_{1:R}$ can be obtained as
\begin{subequations}\label{zwnu}
    \begin{align}
        &z = \frac{1}{\sum_{j=1}^{R}\frac{ \sigma_e^2 \widetilde{y}_{j}^2 \lambda_j^2}{(\lambda_j^2 + 2\nu\sigma_e^2)^2} + \sigma_{{\epsilon}}^2 },\label{znu}\\
        &w_j = \frac{1}{\left( \sum_{j=1}^{R} \frac{ \sigma_e^2 \widetilde{y}_{j}^2 \lambda_j^2}{(\lambda_j^2 + 2\nu\sigma_e^2)^2} + \sigma_{{\epsilon}}^2 \right)}\frac{\widetilde{y}_{j}^2 \lambda_j^2}{(\lambda_j^2 + 2\nu\sigma_e^2)^2},\ j = 1,2,\cdots, R,\label{wnu}
    \end{align}
\end{subequations}
which are functions of $\nu$. 

Note that the primal variables $z$ (\ref{znu}), $\mathbf{w}_{1:R}$ (\ref{wnu}) and the dual variables $\boldsymbol{\eta}_{R+1:N}$ (\ref{etainu}) are related to the dual variable $\nu$.
Substituting (\ref{zwnu}) into (\ref{KKTz}), we obtain a scalar equation $g(\nu) = 0$, where the function $g(\nu)$ is defined as
\begin{align}\label{Rnufunc}
g(\nu) \triangleq &\sum_{j=1}^{R} \left(  \frac{\widetilde{y}_{j}^2 \nu\sigma_e^2}{\lambda_j^2 + 2 \nu \sigma_e^2} \right)  + \frac{1}{2}\|\widetilde{{\mathbf y}}_{R+1:M}\|_2^2 - \frac{M}{2}{\left( \sum_{j=1}^{R} \frac{ \sigma_e^2 \widetilde{y}_{j}^2 \lambda_j^2}{(\lambda_j^2 + 2\nu\sigma_e^2)^2} + \sigma_{{\epsilon}}^2 \right)} \notag\\&+ \nu \sigma_{{\epsilon}}^2.
\end{align}
For $g(\nu)$, we have the following result.
\begin{proposition}\label{prop2}
   $g(\nu)$ is monotonically increasing in the interval $\left(-\frac{\lambda_R^2}{2\sigma_e^2}, \frac{M}{2} \right]$, and a solution such that $g(\nu)=0$ exists.
\end{proposition}
\begin{proof}
It can be seen that $g(\nu)$ has $R$ discontinuity points $\left\{-\frac{\lambda_j^2}{2\sigma_e^2}\right\}_{j=1}^R$, and the most right discontinuity point is $-\frac{\lambda_R^2}{2\sigma_e^2}$. We first compute the first order derivative of $g(\nu)$ as 
\begin{align}
g'(\nu)=\sum_{j=1}^{R} \frac{\widetilde{y}_{j}^2 \lambda_j^2 \sigma_e^2}{(\lambda_j^2 + 2 \nu \sigma_e^2)^2} + 2M \sum_{j=1}^{R} \frac{\sigma_e^4 \widetilde{y}_{j}^2 \lambda_j^2}{(\lambda_j^2 + 2 \nu \sigma_e^2)^3} +  \sigma_{{\epsilon}}^2.
\end{align}
It can be seen that $g'(\nu)>0$ for $\nu>-\frac{\lambda_R^2}{2\sigma_e^2}$, which implies that $g(\nu)$ is monotonically increasing in the interval $(-\frac{\lambda_R^2}{2\sigma_e^2},\infty)$. In addition, $\lim_{\nu\rightarrow -\lambda_R^2/(2\sigma_e^2)^{+}}g(\nu)=-\infty$ and $g\left(\frac{M}{2}\right)$ satisfies 
\begin{align}
    g\left(\frac{M}{2}\right) = \sum_{j=1}^{R} \frac{M^2 \sigma_e^4 \widetilde{y}_{j}^2}{2(\lambda_j^2 + M \sigma_e^2)^2} + \frac{1}{2} \|\widetilde{\mathbf{y}}_{R+1:M}\|_2^2 \geq 0,
\end{align}
where the equality is obtained when $\sigma_e^2=0$ and $R= M$. Because $-\frac{\lambda_R^2}{2\sigma_e^2}<0\leq\frac{M}{2}$ and $g(\nu)$ is monotonically increasing in the interval $(-\frac{\lambda_R^2}{2\sigma_e^2},\frac{M}{2}]$, Proposition \ref{prop2} is established.
\end{proof}

According to (\ref{etainu}), it can be concluded that $\nu > 0$ in this case. This implies that $g(0)>0$. Therefore, the search interval for the optimal value $\nu^\star$ can be further narrowed to $(0,\frac{M}{2}]$.
Now we give the details of the implementation in this case. We use the bisection method to obtain the optimal $\nu^{\star}$ satisfying $g(\nu)=0$. 
Then substituting (\ref{zwnu}) into (\ref{xstar}), one has
\begin{align}\label{xstarwi0}
    &\widetilde{\mathbf x}^\star_{1:R}=\left(\widetilde{\mathbf{y}}_{1:R}\odot \boldsymbol{\lambda}\right)\oslash\left(\boldsymbol{\lambda}^2+2\nu^\star \sigma_e^2\right)=\widetilde{\mathbf{y}}_{1:R}\oslash\left(\boldsymbol{\lambda}+2\nu^\star \sigma_e^2\oslash{\boldsymbol \lambda}\right),\notag\\
    &\widetilde{\mathbf x}^\star_{R+1:N}=0.
\end{align}
% for $j = 1, 2, \ldots, R$, the optimal $\widetilde{x}_j^\star$ is 
% \begin{align}\label{widetildexjML}
%     \widetilde{ x}_j^\star=\widetilde{y}_{j} \lambda_j/\left(\lambda_j^2+2\nu^\star \sigma_e^2\right).
% \end{align}
% And the optimal $\widetilde{x}_i^\star$ for $i = R+1,R+2,\cdots,N$ is $\widetilde{x}_i^\star=0$.
% \begin{align}\label{widetildexML}
%     \widetilde{\mathbf x}^\star=\left(\widetilde{\mathbf{y}}_1\odot \boldsymbol{\lambda}\right)\oslash\left(\boldsymbol{\lambda}^2+2\nu^\star \sigma_e^2\right)=\widetilde{\mathbf{y}}_1\oslash\left(\boldsymbol{\lambda}+2\nu^\star \sigma_e^2\oslash{\boldsymbol \lambda}\right),
% \end{align}
% where $\oslash$ denote componentwise division. 
According to (\ref{definextilde}), $\widehat{\mathbf x}_{\rm ML}$ is 
\begin{align}\label{xMLwi}
    \widehat{\mathbf x}_{\rm ML}={\mathbf V}{\widetilde{\mathbf x}^\star}.
    % =\left({\mathbf H}^{\rm T}{\mathbf H}+2\nu^\star\sigma_e^2{\mathbf I}_N\right)^{-1}{\mathbf H}^{\rm T}{\mathbf y},
\end{align}

\subsubsection{Case $R=N$}\label{subsubcase4}
When $R = N$, $\{w_i,i = R+1,R+2,\cdots,N\}$ is an empty set. Consequently, the primal variables $\mathbf{w}_{R+1:N}$ and the dual variables $\boldsymbol{\eta}_{R+1:N}$ vanish.
And the zero gradient condition (\ref{KKTwi}), the dual constraint (\ref{dualvari}) and the complementary slackness condition (\ref{KKTetai}) are no longer applicable and are therefore removed.

Despite the removed equations, the solutions for the primal variables $\mathbf{w}$ and $z$ in (\ref{zwnu}), as well as the dual variables $\boldsymbol{\eta}$ in (\ref{etaj0wi0}), can still be derived by following the same procedure detailed in Subsubsection \ref{subsubcase3}, specifically from (\ref{sumnowi}) to (\ref{zwnu}).

Substituting (\ref{zwnu}) into (\ref{xstar}), the optimal $\widetilde{\mathbf x}^\star$ can be obtained as 
\begin{align}\label{widetildexML}
    \widetilde{\mathbf x}^\star=\left(\widetilde{\mathbf{y}}\odot \boldsymbol{\lambda}\right)\oslash\left(\boldsymbol{\lambda}^2+2\nu^\star \sigma_e^2\right)=\widetilde{\mathbf{y}}\oslash\left(\boldsymbol{\lambda}+2\nu^\star \sigma_e^2\oslash{\boldsymbol \lambda}\right).
\end{align} 
$\widehat{\mathbf x}_{\rm ML}$ is 
\begin{align}\label{xML}
    \widehat{\mathbf x}_{\rm ML}={\mathbf V}{\widetilde{\mathbf x}^\star}=\left({\mathbf H}^{\rm T}{\mathbf H}+2\nu^\star\sigma_e^2{\mathbf I}_N\right)^{-1}{\mathbf H}^{\rm T}{\mathbf y}.
\end{align} 
% Whether $\mathbf{H}$ is rank-deficient ($R < \operatorname{min}(M,N)$) or full-rank ($R = \operatorname{min}(M,N)$), the KKT conditions ultimately collapse to the same one-dimensional convex equation in $\nu$, which indicates the generality of our result.
Note that the form of $\widehat{\mathbf x}_{\rm ML}$ (\ref{xML}) is the same as that of TLS and LS, except that $2\nu\sigma_e^2$ is replaced with $\alpha_{\rm TLS}$ and $0$ \cite{EldarCRB}.

To better understand the consistency between the special case $R = N$ in Subsubsection \ref{subsubcase2} and Subsubsection \ref{subsubcase4}, we calculate $g(0)$ as
\begin{align}\label{g0}
    g(0) = \frac{1}{2} \|\widetilde{{\mathbf y}}_{R+1:M}\|_2^2  - \frac{M}{2} \left( \sum_{j=1}^{R} \frac{\sigma_e^2 \widetilde{y}_{j}^2}{\lambda_j^2} + \sigma_{{\epsilon}}^2 \right).
\end{align}
Note that when (\ref{g0}) equals zero, it implies that $\nu = 0$, indicating a special case in which the dual variable associated with the primary constraint vanishes. Notably, $g(0)$ (\ref{g0}) and $S$ (\ref{sumgeq0}) satisfies 
\begin{align}\label{Sg0}
    S=\frac{g(0)}{\frac{1}{2} \|\widetilde{{\mathbf y}}_{R+1:M}\|_2^2\sigma_e^2}.
\end{align}
Substituting $\nu = 0$ into (\ref{widetildexML}), $\widetilde{\mathbf{x}}^\star$ (\ref{widetildexML}) is consistent with (\ref{subcase2xstar}) shown in Subsubsection \ref{subsubcase2}.
This makes sense as the analysis is carried out under the condition $R=N$. It can be observed that although this subsection focuses on the case where $\nu \neq 0$, the resulting conclusions are consistent with that obtained when $\nu = 0$. Therefore, this scenario extends the result derived in Subsubsection \ref{subsubcase2} to a more general one.

% Based on the theoretical derivations in Subsections \ref{subcasenu0} and \ref{subcasenunot0}, the solution to the ML problem can be categorized into four distinct cases depending on $R=\operatorname{rank}(\mathbf{H})$ and $S$ (\ref{sumgeq0}). If $R \leq N - 1$ and $S \geq 0$, the conditions described in Subsubsection \ref{subsubcase1} are satisfied, and the optimal solution to the ML problem can be obtained using  (\ref{xstar}) and (\ref{xMLend}). If $R \leq N - 1$ and $S < 0$, the scenario falls under the conditions outlined in Subsubsection \ref{subsubcase3}, where (\ref{xMLwi}) provides the optimal solution to the ML problem. \blue{When $R = N$ and $S = 0$ where Subsubsection \ref{subsubcase2} applies}, and the ML solution can be derived from  (\ref{subcase2xstar}) and (\ref{xMLend}). Finally, if $R = N$ and $S \neq 0$, Subsubsection \ref{subsubcase4} applies, and the optimal solution can be determined using (\ref{xML}).

Based on the theoretical derivations in Subsections \ref{subcasenu0} and \ref{subcasenunot0}, the solution to the ML problem can be categorized into four distinct cases depending on $R=\operatorname{rank}(\mathbf{H})$ and $S$ (\ref{sumgeq0}).
When $R \leq N - 1$ and $S \geq 0$, the conditions in Subsubsection \ref{subsubcase1} are satisfied, and the  ML solution is obtained by (\ref{xstar}) and (\ref{xMLend}). If $R \leq N - 1$ and $S < 0$, the scenario falls under the conditions outlined in Subsubsection \ref{subsubcase3}, where (\ref{xMLwi}) provides the optimal solution to the ML problem.
In the case where $R = N$ and $S = 0$, as described in Subsubsection \ref{subsubcase2}, the ML solution can be obtained from (\ref{subcase2xstar}) and (\ref{xMLend}). Finally, when $R = N$ and $S \neq 0$, Subsubsection \ref{subsubcase4} applies, and the optimal solution is determined using (\ref{xML}).

According to (\ref{Sg0}), the sign of $S$ is consistent with the sign of $g(0)$. Therefore, when solving $g(\nu^{\star}) = 0$ using the bisection method, the sign of $S$ can be utilized to narrow the search interval. Specifically, if $S > 0$, which implies $g(0) > 0$, the interval can be reduced to $\left(-\frac{\lambda_R^2}{2\sigma_e^2}, 0 \right)$. Conversely, if $S < 0$, implying $g(0) < 0$, the interval can be narrowed to $\left(0, \frac{M}{2} \right]$.

Following the analysis, the overall solution procedure proceeds by first computing the rank $R$ and the scalar $S$, then a corresponding closed-form or numerical solution is used to compute the ML estimate, which is named as GRV-ML. The implementation steps of GRV-ML are summarized in Algorithm \ref{algor}.

\begin{algorithm}
\renewcommand{\algorithmicrequire}{\textbf{Input:}}
\caption{GRV-ML}
\label{algor}
\begin{algorithmic}[1]
\REQUIRE  $\mathbf H$, $\mathbf{y}$, $\sigma_e^2$ and $\sigma_{{\epsilon}}^2$
    \STATE $[\mathbf U,\boldsymbol\Sigma,\mathbf V] = \operatorname{svd}(\mathbf H)$, $R=\operatorname{rank}(\mathbf{H})$
    \STATE $\boldsymbol{\lambda} = \operatorname{diag}(\boldsymbol\Sigma_{1:R,1:R})$, where $\boldsymbol{\Sigma}_{1:R,1:R}$ denotes the first $R$ rows and columns of $\boldsymbol{\Sigma}$, and
    $ \lambda_1 \geq \lambda_2 \geq \dots \geq \lambda_R > 0 $.
    \STATE $\widetilde{\mathbf{y}}={{\mathbf{U}}^{\text{T}}}\mathbf{y}$
    \IF{$R<M$}
    \STATE $S=\frac{g(0)}{\frac{1}{2} \|\widetilde{{\mathbf y}}_{R+1:M}\|_2^2\sigma_e^2}$\label{cals}
    \ELSE
    \STATE $R=M$
    \STATE $S = -\infty$
    \ENDIF
    \IF{$R\leq N-1$}
    \IF{$S\geq 0$}
    \STATE $\nu=0$\label{stepcase1}
    \STATE $z^\star=\frac{M}{\|\widetilde{{\mathbf y}}_{R+1:M}\|_2^2}$
    \STATE \text{Choose} $\mathbf w_{R+1:N}^\star \geq 0$ \text{such that} $\sum_{i=R+1}^{N} w_i^\star = S$.
    \STATE $\widetilde{\mathbf x}^\star_{1:R}=\widetilde{\mathbf y}_{1:R}\oslash \boldsymbol{\lambda}$
    \STATE$\widetilde{\mathbf x}^\star_{R+1:N}=\pm \sqrt{\mathbf w_{R+1:N}^\star/  z^\star}$\label{stepcase1end}
    \ELSE
    \STATE Find $\nu^\star$ satisfying $g(\nu^\star) = 0$ via bisection method in the interval $\left(0, \frac{M}{2} \right]$.\label{stepcase2}
    \STATE $\widetilde{\mathbf x}^\star_{1:R}=\widetilde{\mathbf{y}}_{1:R}\oslash\left(\boldsymbol{\lambda}+2\nu^\star \sigma_e^2\oslash{\boldsymbol \lambda}\right)$
    \STATE $\widetilde{\mathbf x}^\star_{R+1:N}=0$\label{stepcase2end}
    \ENDIF
    \ELSE
    \IF{$S=0$}
    \STATE $\widetilde{\mathbf x}^\star=\widetilde{\mathbf y}\oslash \boldsymbol{\lambda}$\label{stepcase3}
    \ELSIF{$S>0$}
    \STATE Find $\nu^\star$ satisfying $g(\nu^\star) = 0$ via bisection method in the interval $\left(-\frac{\lambda_R^2}{2\sigma_e^2}, 0 \right)$.\label{stepcase4}
    \STATE $\widetilde{\mathbf x}^\star=\widetilde{\mathbf{y}}\oslash\left(\boldsymbol{\lambda}+2\nu^\star \sigma_e^2\oslash{\boldsymbol \lambda}\right)$
    \ELSE 
    \STATE Find $\nu^\star$ satisfying $g(\nu^\star) = 0$ via bisection method in the interval $\left(0, \frac{M}{2} \right]$.
    \STATE $\widetilde{\mathbf x}^\star=\widetilde{\mathbf{y}}\oslash\left(\boldsymbol{\lambda}+2\nu^\star \sigma_e^2\oslash{\boldsymbol \lambda}\right)$\label{stepcase4end}
    \ENDIF
    \ENDIF
    \STATE \textbf{return} $\widehat{\mathbf x}_{\rm ML}={\mathbf V}{\widetilde{\mathbf x}^\star}$
\end{algorithmic}
\end{algorithm}

We analyze the computation complexity of the proposed algorithm, the method \cite{SPLeldar} and the EM algorithm \cite{EldarCRB}. Let ${\rm Iter}$ denote the number of iterations in bisection methods of the proposed algorithm and the method \cite{SPLeldar}, and the number of iterations of the EM algorithm \cite{EldarCRB}. 
Since \cite{SPLeldar} and \cite{EldarCRB} consider only the full column rank case, i.e., $R = N$, their methods correspond to the case discussed in Subsubsections \ref{subsubcase2} and \ref{subsubcase4}, which are covered by Steps \ref{stepcase3}-\ref{stepcase4end} in Algorithm \ref{algor}. Moreover, as $S$ is a function of the continuous random variable $\widetilde{\mathbf{y}}$, the probability of $S = 0$ is zero. Therefore, in the comparison, we restrict our calculation to the solution of $\widetilde{\mathbf{x}}$ using Steps \ref{stepcase4}-\ref{stepcase4end} of Algorithm \ref{algor}.
For the EM algorithm, the computation complexity is dominated by $\overline{\mathbf{G}_N^{\rm T}\mathbf{G}_N}$ and $\overline{\mathbf{G}_N}\in\mathbb{R}^{M\times N}$ for each iteration, which takes $\mathcal O(MN^2\rm{Iter})$, where $\overline{\mathbf{G}_N}$ is the estimate of the measurement matrix. For the method \cite{SPLeldar}, the computation complexity is dominated by the matrix multiplication  $\left(\mathbf{H}^{\rm T}  \mathbf{H}+\alpha \mathbf{I}\right)^{\dagger} \mathbf{H}^{\rm T} \mathbf{y}$, which can be  simplified by using an economy SVD for $\mathbf{H}^{\rm T} \mathbf{H}$ and takes $\mathcal O(MN^2+N^2\rm{Iter})$\cite{matrixcomp}. For the proposed algorithm, only an economy SVD is needed as $\|\widetilde{\mathbf y}_{R+1:M}\|_2^2$ can be computed via $\|\widetilde{\mathbf y}_{R+1:M}\|_2^2=\|\widetilde {\mathbf y}\|_2^2-\|\widetilde{\mathbf y}_{1:R}\|_2^2$. The computation consists of an economy SVD step, a bisection method to obtain $\nu^{\star}$, and the recovery of $\widehat{\mathbf x}_{\rm ML}$, which takes $\mathcal O(MN^2)$, $\mathcal O(N{\rm Iter})$ and $\mathcal O(MN)$, respectively, and the overall computation complexity is $\mathcal O(MN^2+N{\rm Iter})$, which is far more less than the above two approaches.

\section{Numerical Simulation}
In this section, several numerical experiments are conducted to validate the theoretical results of the proposed method under general settings and illustrate the performance of the ML estimator.

\begin{figure*}
    \centering
    \subfigure[]{
    \label{MLEfunction}
    \includegraphics[width = 2.3in]{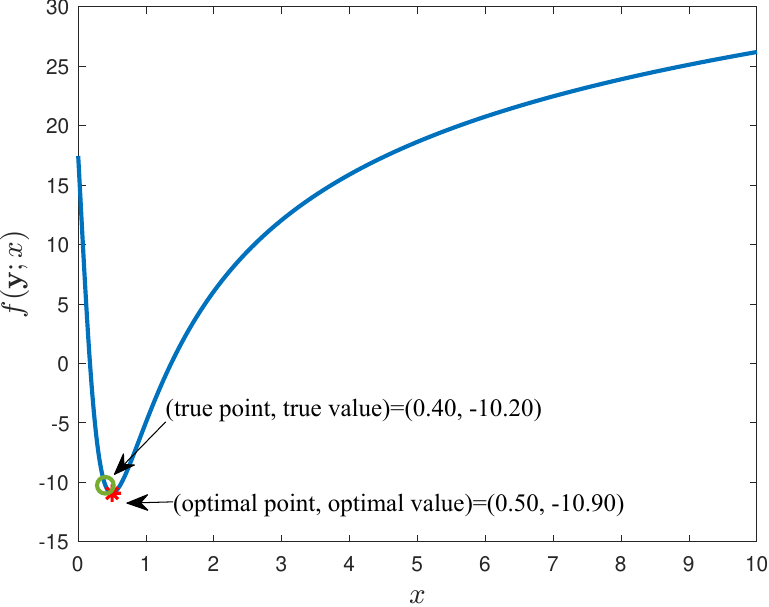}}
    \subfigure[]{
    \label{convexverify}
    \includegraphics[width = 2.6in]{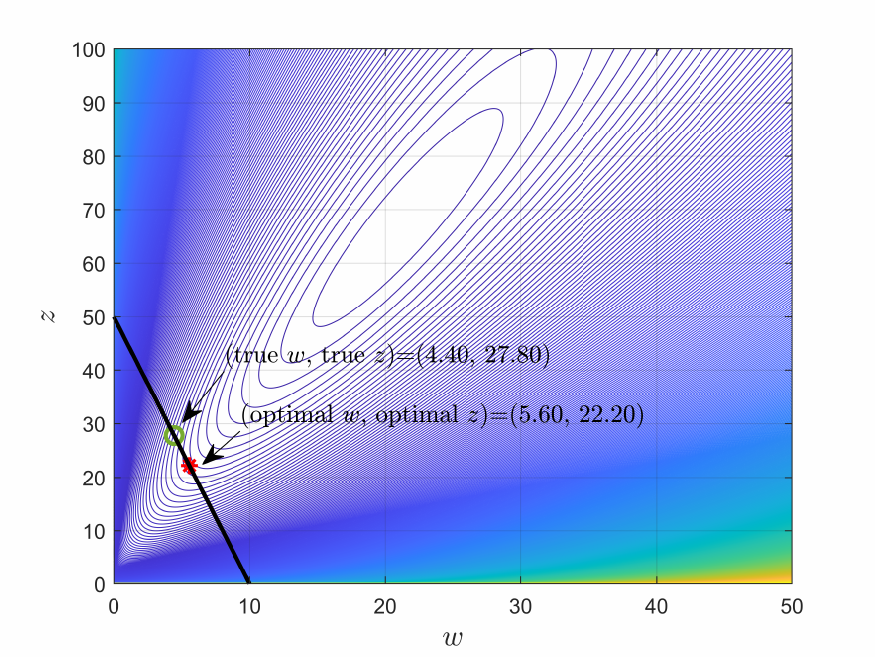}}
	\caption{An example illustrating how we transform the original unconstrained quasi-convex problem to a constrained convex problem by using the lifting technique. 
    (a): the MLE function $f({\mathbf y};x)$ with respect to $x$ which is a quasi-convex function. (b): contour plot of the constrained optimization function after the lifting technique. The black line represents the linear constraint, the red star represents the optimal point and the optimal value and the green circle represents the true value.}\label{liftfigure}
\end{figure*}

\subsection{Convex Verification}
In this subsection, we use a simple example to illustrate how the original quasiconvex objective function is transformed to a linear equality constrained convex optimization problem, as shown in Fig. \ref{liftfigure}. Let $N=1$, $M=4$, ${\mathbf H}=[0.8084,0.7673,0.6168,0.5360]^{\rm T}$, $\sigma_e^2=0.10$, $\sigma_{{\epsilon}}^2=0.02$, $x=0.40$ and $V=1$ is the right singular vector of ${\mathbf H}$. Fig. \ref{MLEfunction} shows the objective function, which is obviously nonconvex but quasiconvex. The optimal point $\widehat{x}_{\rm ML}$ and the optimal value are $0.50$ and $-10.90$, respectively. 

For our modeling, we use the lifting technique and the number of variables is $2$, i.e., $w$ and $z$. Fig. \ref{convexverify} presents the contour lines of the objective function.
The increasing density of contour lines visually suggests a sharper increase of the transformed MLE function, which can be viewed as a indicator of convexity. 
The optimality condition at the red star can be verified geometrically: The line $\sigma _{e}^{2} w+\sigma_{{\epsilon}}^2 z=1$ is tangent to the contour line of the lifted function through $(w^{\star},z^{\star})$.
The optimal point is $(w^{\star},z^{\star})=(5.60,22.20)$. The optimal value is -10.90, which equals to the optimal value of the original problem shown in Fig. \ref{MLEfunction}. In this way, $x^{\star}$ can be calculated by $x^{\star}=\operatorname{sign}(y_1)V\sqrt{w^{\star}/z^{\star}}\approx 0.50$. It can be seen that these two optimization problems are equivalent.

\subsection{Examples for the Four Cases}
According to the theoretical analysis in Subsections \ref{subcasenu0} and \ref{subcasenunot0}, the solution to the ML problem (\ref{mle}) can be categorized into four distinct cases: $R \leq N - 1$ and $S \geq 0$ (Sec. \ref{subsubcase1}); $R \leq N - 1$ and $S < 0$ (Sec. \ref{subsubcase3}); $R = N$ and $S = 0$ (Sec. \ref{subsubcase2}); and $R = N$ with $S \neq 0$ (Sec. \ref{subsubcase4}). Among these, cases when $R \leq N - 1$ may correspond to either underdetermined or overdetermined scenarios, whereas cases when $R = N$ can only occur in the overdetermined setting.
In this subsection, examples for these cases are presented to validate the correctness of the ML solution derived from the KKT conditions. As the expression for $S$ in Eq. (\ref{sumgeq0}) involves continuous random variables, the probability of $S=0$ is zero. Thus the case in which $S=0$ is not presented in the examples.

\subsubsection{Example 1 where $R \leq N - 1$ and $S\geq0$}
The first example shows the case where $R \leq N - 1$ and $S\geq0$, as discussed in Subsubsection \ref{subsubcase1}.
% The first experiment considers the overdetermined case where $M > N$. 
The parameters are set as follows: $M = 4$, $N = 2$, $\mathbf{x} = [0.33, 0.62]^{\rm T}$, $\mathbf{y} = [-0.16, 0.28, -0.30,-1.17]^{\rm T}$, and
\begin{align}
\mathbf{H} = 
\begin{bmatrix}
-0.44 & -0.43 \\
0.48 & 0.46 \\
-0.89 & -0.85 \\
-1.07 & -1.03 \\
\end{bmatrix},
\end{align}
with $R =\operatorname{rank}(\mathbf{H}) = 1= N - 1<N=2$, $\sigma_e^2 = 0.10$, and $\sigma_{{\epsilon}}^2 = 0.03$. According to Eq. (\ref{sumgeq0}), $S$ is calculated to be $S = 0.2175$. Because $R\leq N-1$ and $S>0$, Eq. (\ref{xstar}) can be used to find the optimal solution  $\widetilde{\mathbf{x}}^\star$, and the number of optimal solutions is $2$. The optimal solutions $\widetilde{\mathbf{x}}^\star$ are $[0.53,0.41]^{\rm T}$ and $[0.53,-0.41]^{\rm T}$.

According to $\widehat{\mathbf x}_{\rm ML}={\mathbf V}{\widetilde{\mathbf x}^\star}$, the ML estimates are $\widehat{\mathbf x}_{\rm{ML}1} = [0.66, 0.07]^{\rm T}$ and $\widehat{\mathbf x}_{\rm{ML}2} = [0.10, 0.66]^{\rm T}$. 
% Substituting each solution into the negative loglikelihood function (\ref{fxydef}) results in the same optimal value of $-6.3027$. These results are consistent with the analysis in Section III-A1, which states that the sign of $\widetilde{\mathbf{x}}^\star_{R+1:N}$ is undetermined in (\ref{xstar}) when $\nu = 0$, leading to two optimal solutions. Specifically, in this experiment, we have $\widetilde{x}_2^\star = \pm 0.4068$, corresponding to the two ML estimators.
\begin{figure}
	\centering
	\includegraphics[width = 4in]{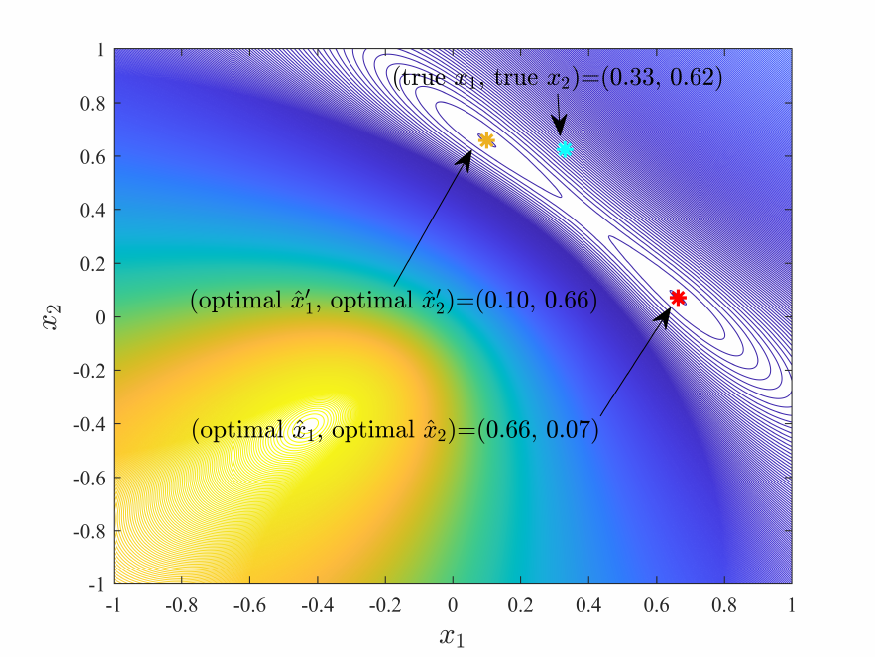}
	\caption{Example 1: 
    The contour plot of $f({\mathbf y};\mathbf x)$ with respect to the variables $x_1$ and $x_2$. The red and yellow stars denote the optimal solutions $\widehat{\mathbf x}_{\rm{ML}1} = [0.66, 0.07]^{\rm T}$ and $\widehat{\mathbf x}_{\rm{ML}2} = [0.10, 0.66]^{\rm T}$, respectively, and the blue star denotes the ground truth $\mathbf{x} = [0.33, 0.62]^{\rm T}$.}\label{Example1}
\end{figure}
In order to validate the correctness of the calculated ML estimates, the contour plot of the negative loglikelihood function $f({\mathbf y};\mathbf x)$ in (\ref{fxydef}) with respect to the variables $x_1$ and $x_2$ is created, as shown in Fig. \ref{Example1}. 

The red and yellow stars denote the two optimal solutions $\widehat{\mathbf x}_{\rm{ML}1} = [0.66, 0.07]^{\rm T}$ and $\widehat{\mathbf x}_{\rm{ML}2} = [0.10, 0.66]^{\rm T}$, while the blue star indicates the ground truth. As shown in Fig. \ref{Example1}, the ML objective achieves the same global minimum at two distinct points, which verifies the theoretical analysis in Subsubsection \ref{subsubcase1} and confirms the existence of multiple optimal solutions under this setting.

\subsubsection{Example 2 where $R \leq N - 1$ and $S<0$}
Example 2 shows the case where $R \leq N - 1$ and $S<0$, as discussed in Subsubsection \ref{subsubcase3} and this example considers the overdetermined case where $M>N$.
% The first experiment considers the overdetermined case where $M > N$. 
The parameters are set as follows: $M = 4$, $N = 2$, $\mathbf{x} = [0.62,0.23]^{\rm T}$, $\mathbf{y} = [-0.13,0.55,0.01,0.40]^{\rm T}$, and
\begin{align}
\mathbf{H} = 
\begin{bmatrix}
0.14 & -1.16\\
-0.23 & 1.86 \\
0.04 & -0.33 \\
-0.40 & 3.23
\end{bmatrix},
\end{align}
with $R =\operatorname{rank}(\mathbf{H}) = 1= N - 1<N=2$, $\sigma_e^2 = 0.10$, and $\sigma_{{\epsilon}}^2 = 0.03$. According to Eq. (\ref{sumgeq0}), $S$ is calculated to be $S = -6.00$. Because $R\leq N-1$ and $S<0$, Eq. (\ref{xstarwi0}) can be used to find the optimal solution  $\widetilde{\mathbf{x}}^\star$, and $\nu^\star$ can be efficiently calculated using the bisection method within the interval $(0, M/2]$ according to Proposition \ref{prop2}. Substituting  $\nu^\star=0.75$ into (\ref{xstarwi0}), the optimal solutions $\widetilde{\mathbf{x}}^\star$ is $[0.16,0.00]^{\rm T}$.

According to $\widehat{\mathbf x}_{\rm ML}={\mathbf V}{\widetilde{\mathbf x}^\star}$, the ML estimate is $\widehat{\mathbf x}_{\rm{ML}} = [-0.02,0.16]^{\rm T}$. 
\begin{figure}
	\centering
	\includegraphics[width = 4in]{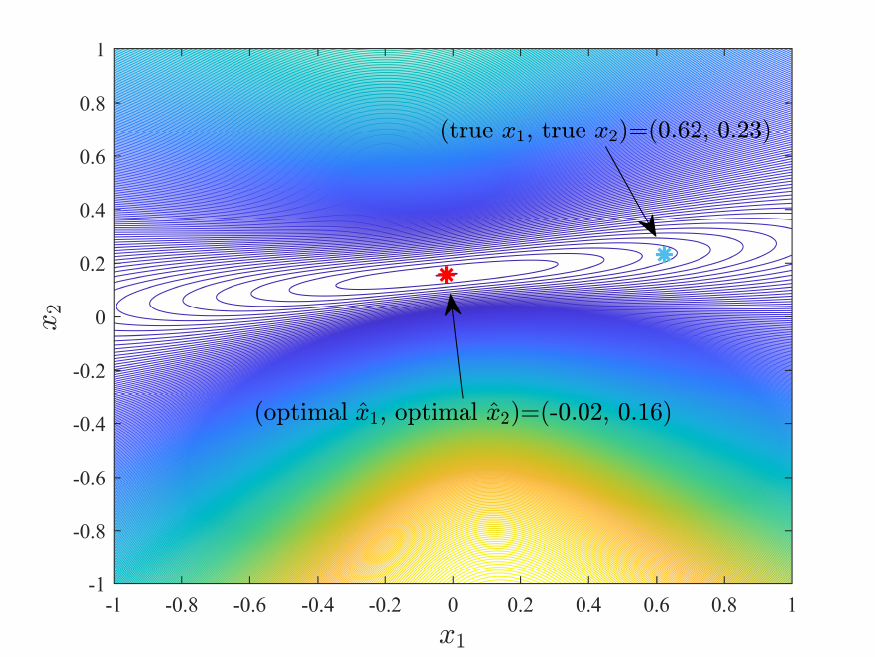}
	\caption{Example 2: 
    The contour plot of $f({\mathbf y};\mathbf x)$ with respect to the variables $x_1$ and $x_2$. The red star denotes the optimal solution $\widehat{\mathbf x}_{\rm{ML}} = [-0.02,0.16]^{\rm T}$, and the blue star denotes the ground truth $\mathbf{x} = [0.62,0.23]^{\rm T}$.}\label{Example2}
\end{figure}
To validate the correctness of the calculated ML estimate, Fig. \ref{Example2} presents the contour plot of the negative loglikelihood function $f({\mathbf y};\mathbf x)$ in (\ref{fxydef}) with respect to the variables $x_1$ and $x_2$. 

The red star denotes the optimal solution $\widehat{\mathbf x}_{\rm{ML}} = [-0.02,0.16]^{\rm T}$, and the blue star indicates the ground truth. As shown in Fig. \ref{Example2}, the ML objective achieves the global minimum at the red point, which verifies the theoretical analysis in Subsubsection \ref{subsubcase3} when $M>N$.

\subsubsection{Example 3 where $R =M\leq N - 1$ and $S=-\infty$}
Example 3 shows the case where $R \leq N - 1$ and $S=-\infty$, as discussed in Subsubsection \ref{subsubcase3} and this example considers the underdetermined case where $M<N$. The parameters are set as follows: $M = 1$, $N = 2$, $y = 2.86$, $\mathbf{x} = [0.77,0.60]^{\rm T}$, and
$\mathbf{H} = [1.93,2.61]$
with $R =\operatorname{rank}(\mathbf{H})=1= N - 1<N = 2$, $\sigma_e^2 = 0.10$, and $\sigma_{{\epsilon}}^2 = 0.03$. According to Eq. (\ref{sumgeq0}), $S$ is calculated to be $S = -\infty$. Because $R\leq N-1$ and $S<0$, Eq. (\ref{xstarwi0}) can be used to find the optimal solution $\widetilde{\mathbf{x}}^\star$, and $\nu^\star$ can be efficiently calculated using the bisection method within the interval $(0, M/2]$ according to Proposition \ref{prop2}. Substituting  $\nu^\star=0.50$ into (\ref{xstarwi0}), the optimal solution $\widetilde{\mathbf{x}}^\star$ is $[0.87,0.00]^{\rm T}$.

According to $\widehat{\mathbf x}_{\rm ML}={\mathbf V}{\widetilde{\mathbf x}^\star}$, the ML estimate is $\widehat{\mathbf x}_{\rm{ML}} = [0.52, 0.70]^{\rm T}$. 
To validate the correctness of the calculated ML estimate, the contour plot of the negative loglikelihood function $f({\mathbf y};\mathbf x)$ (\ref{fxydef}) with respect to the variables $x_1$ and $x_2$ is created, as shown in Fig. \ref{Example3}. 
\begin{figure}
	\centering
	\includegraphics[width = 4in]{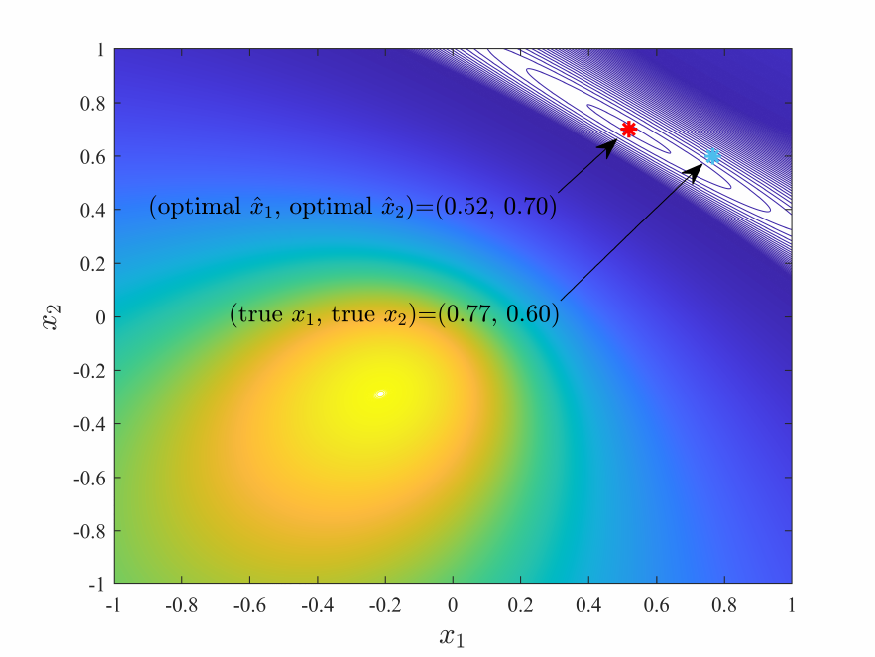}
	\caption{Example 3: The contour plot of $f({\mathbf y};\mathbf x)$ with respect to the variables $x_1$ and $x_2$. The red star denotes the optimal solution $\widehat{\mathbf x}_{\rm{ML}} = [0.52, 0.70]^{\rm T}$, and the blue star denotes the ground truth $\mathbf{x} = [0.77,0.60]^{\rm T}$.}\label{Example3}
\end{figure}

The red star marks the optimal solution $\widehat{\mathbf x}_{\rm{ML}} = [0.52, 0.70]^{\rm T}$, while the blue star indicates the ground truth. This example illustrates the estimation performance in the underdetermined case, where the number of measurements $M$ is strictly less than the number of unknown variables $N$. In this scenario, the ML estimate remains unique and shows that the randomness in the measurement matrix can play a beneficial role in the underdetermined case. As shown in Fig. \ref{Example3}, the solution obtained according to Subsubsection \ref{subsubcase3} successfully identifies this unique optimum.

\subsubsection{Example 4 where $R =N$ and $S<0$}
Example 4 shows the case where $R =N$ and $S<0$, as discussed in Subsubsection \ref{subsubcase4}. The parameters are set as follows: $M = 4$, $N = 2$, $\mathbf{x} = [0.61,0.42]^{\rm T}$, $\mathbf y = [-0.85,0.45,-1.27,-0.82]^{\rm T}$, and
\begin{align}
\mathbf{H} = 
\begin{bmatrix}
-1.92 & -0.05 \\
0.74 & -0.04 \\
-2.36 & -0.66 \\
-0.61 & -0.30 \\
\end{bmatrix},
\end{align}
with $R = \operatorname{rank}(\mathbf{H})=2=N$, $\sigma_e^2 = 0.10$, and $\sigma_{{\epsilon}}^2 = 0.03$.

According to Eq.(\ref{sumgeq0}), $S$ is calculated to be $S = -10.07<0$. Because $R= N$ and $S<0$, Eq. (\ref{widetildexML}) can be used to find the optimal solution $\widetilde{\mathbf{x}}^\star$, and $\nu^\star$ can be efficiently calculated using the bisection method within the interval $(0, M/2]$ according to Proposition \ref{prop2}. Substituting  $\nu^\star=0.64$ into (\ref{widetildexML}), the optimal solution $\widetilde{\mathbf{x}}^\star$ is $[0.53,0.35]^{\rm T}$.

According to (\ref{xML}), the ML estimate is $\widehat{\mathbf x}_{\rm{ML}} = [0.45, 0.44]^{\rm T}$. 
To assess the accuracy of the calculated ML estimates, Fig. \ref{Example4} displays the contour plot of the negative loglikelihood function $f({\mathbf y};\mathbf x)$ in (\ref{fxydef}), plotted over the variables $x_1$ and $x_2$.

\begin{figure}
	\centering
	\includegraphics[width = 4in]{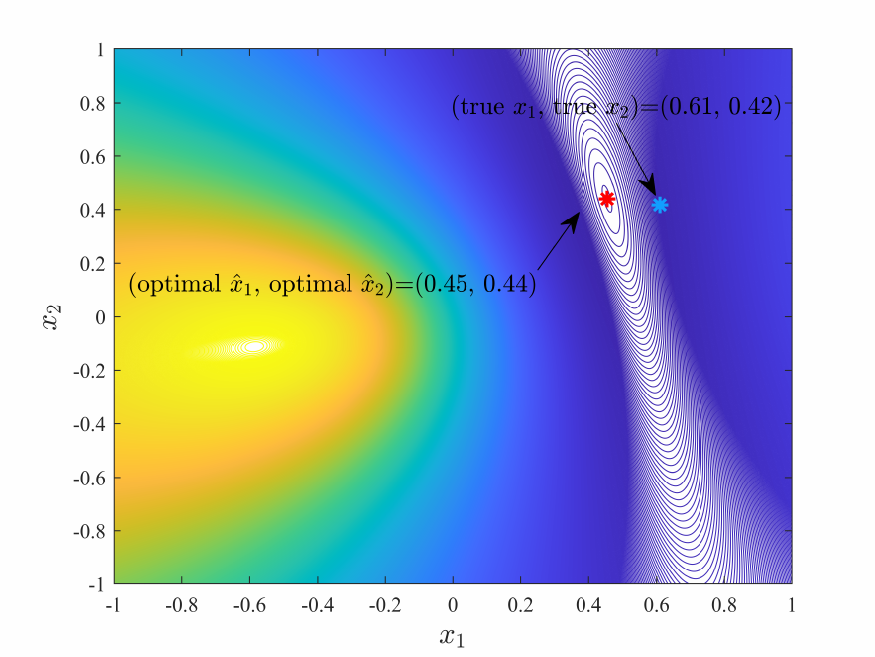}
	\caption{Example 4: The contour plot of $f({\mathbf y};\mathbf x)$ with respect to the variables $x_1$ and $x_2$. The red star denotes the optimal solution $\widehat{\mathbf x}_{\rm{ML}} = [0.45, 0.44]^{\rm T}$, and the blue star denotes the ground truth $\mathbf{x} = [0.61,0.42]^{\rm T}$.}\label{Example4}
\end{figure}
The optimal solution $\widehat{\mathbf x}_{\rm{ML}} = [0.45, 0.44]^{\rm T}$ is marked with a red star, while the ground truth is indicated by a blue star. As illustrated in Fig. \ref{Example4}, the method described in Subsubsection \ref{subsubcase4} successfully identifies the global optimum of the ML problem in this experiment, further validating the correctness of the theoretical result.

\subsubsection{Example 5 where $R =N$ and $S>0$}
Example 5 shows the case where $R =N$ and $S>0$, as discussed in Subsubsection \ref{subsubcase4}. The parameters are set as follows: $M = 4$, $N = 2$, $\mathbf{x} = [0.51,0.69]^{\rm T}$, $\mathbf y = [0.06,-0.20,0.20,-0.48]^{\rm T}$, and
\begin{align}
\mathbf{H} = 
\begin{bmatrix}
-0.40 & 0.71\\
0.64 & -0.64\\
2.65 & -1.85\\
-1.39 & 1.43\\
\end{bmatrix},
\end{align}
with $R = \operatorname{rank}(\mathbf{H})=2=N$, $\sigma_e^2 = 0.10$, and $\sigma_{{\epsilon}}^2 = 0.03$.
According to Eq. (\ref{sumgeq0}), $S$ is calculated to be $S = 2.87>0$. Because $R= N$ and $S>0$, Eq. (\ref{widetildexML}) can be used to find the optimal solution $\widetilde{\mathbf{x}}^\star$, and $\nu^\star$ can be efficiently calculated using the bisection method within the interval $(-\frac{\lambda_R^2}{2\sigma_e^2}, 0)$ according to Proposition \ref{prop2}. Substituting  $\nu^\star= -0.48$ into (\ref{widetildexML}), the optimal solution $\widetilde{\mathbf{x}}^\star$ is $[0.09,-0.27]^{\rm T}$.

According to Eq. (\ref{xML}), the ML estimate is $\widehat{\mathbf x}_{\rm{ML}} = [-0.10,-0.26]^{\rm T}$. 
To validate the correctness of the calculated ML estimates, the contour plot of the negative loglikelihood function $f({\mathbf y};\mathbf x)$ (\ref{fxydef}) with respect to the variables $x_1$ and $x_2$ is created, as shown in Fig. \ref{Example5}.

\begin{figure}
	\centering
	\includegraphics[width = 4in]{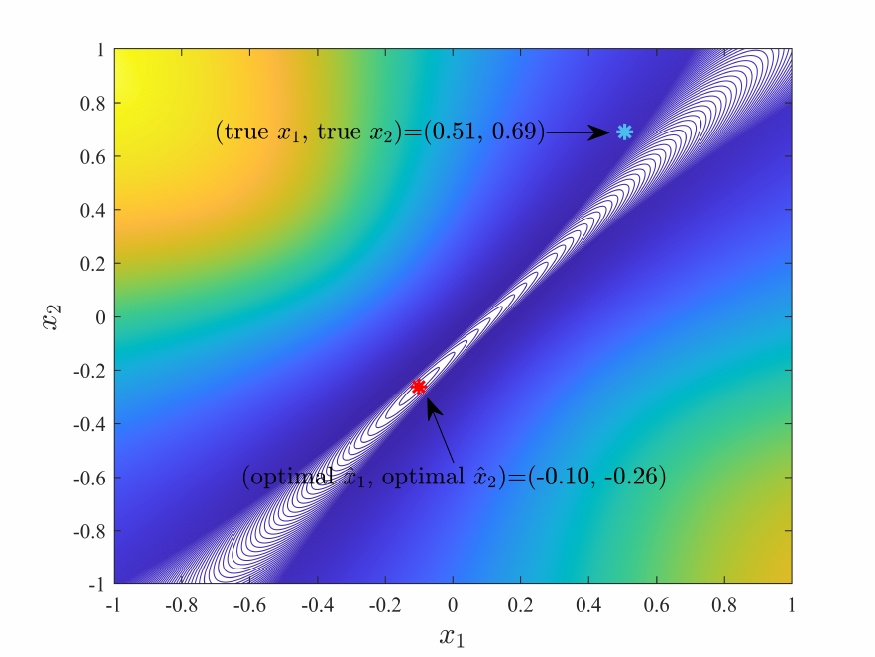}
	\caption{Example 5: The contour plot of $f({\mathbf y};\mathbf x)$ with respect to the variables $x_1$ and $x_2$. The red star denotes the optimal solution $\widehat{\mathbf x}_{\rm{ML}} = [-0.10,-0.26]^{\rm T}$, and the blue star denotes the ground truth $\mathbf{x} = [0.51,0.69]^{\rm T}$.}\label{Example5}
\end{figure}

The red star denotes the optimal solution $\widehat{\mathbf x}_{\rm{ML}} = [-0.10,-0.26]^{\rm T}$, and the blue star represents the true value. As illustrated in Fig. \ref{Example5}, the method described in Subsubsection \ref{subsubcase4} successfully identifies the global optimum of the ML problem in this experiment.
Note that although this example is overdetermined, the estimation error remains large when $\nu<0$.

\subsection{Statistical Evaluation in Underdetermined Settings}

\begin{figure}
	\centering
	\includegraphics[width = 4in]{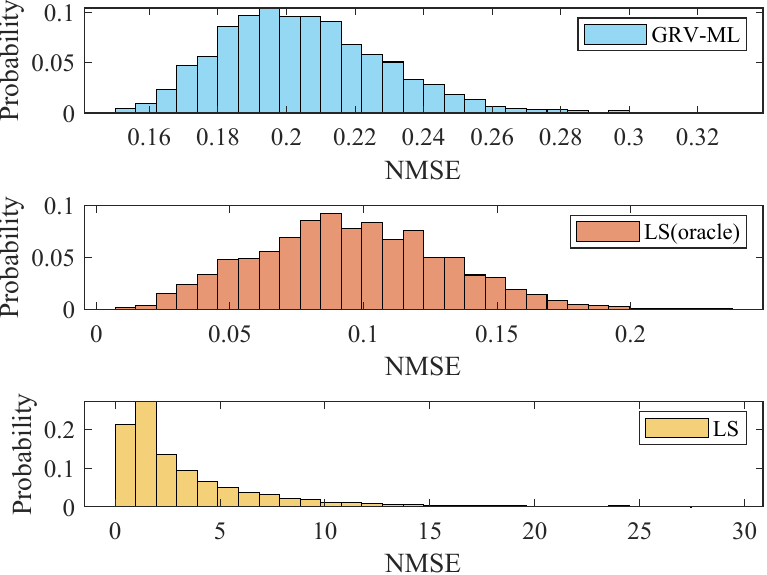}
	\caption{The histogram of the NMSE distributions in estimating $\mathbf{x}$ of RV model.}\label{NMSE}
\end{figure}
In this subsection, we statistically evaluate the estimation error of Algorithm \ref{algor}, the LS given by ${\mathbf H}^{\dagger}\mathbf{y}$, and the LS (oracle), which assumes that $\mathbf E$ is perfectly known and is obtained by $({\mathbf H}+{\mathbf E})^{\dagger}\mathbf{y}$. The evaluation is conducted in the underdetermined case where $M < N$ and $R = \operatorname{rank}(\mathbf{H}) = M$. The parameters are set as follows: $M = 95$, $N = 100$, $\sigma_e^2 = 0.10$, and $\sigma_\epsilon^2 = 0.03$. $\mathbf H$ and $\mathbf x$ are generated and fixed across 2000 Monte Carlo (MC) trials where $H_{ij}\sim {\mathcal N}(0,1)$ and $\mathbf x\sim {\mathcal N}(\mathbf 0,I_N)$. According to Eq. (\ref{sumgeq0}), $S$ is calculated to be $S = -\infty$. Because $R\leq N-1$ and $S< 0$, the settings corresponds to the case discussed in Subsubsection \ref{subsubcase3}. In this scenario, a unique solution can be obtained by following Steps \ref{stepcase2}-\ref{stepcase2end} of Algorithm \ref{algor}.
For an estimator $\widehat{\mathbf x}$, its normalized mean squared error (NMSE) is defined as $\text{NMSE}(\widehat{\mathbf x})=\frac{\|\widehat{\mathbf x}-\mathbf x\|_2^2}{\|\mathbf x\|_2^2}$.
The histogram of NMSE distributions is plotted in Fig. \ref{NMSE} to show each possible state and its probability distribution.

As observed, the NMSE of the proposed GRV-ML is concentrated around 0.19, while the LS (oracle) solution achieves a lower NMSE centered at approximately 0.09. In contrast, the LS estimator yields a significantly larger NMSE, centered around 1.95. 
Notably, in the underdetermined settings, the ML estimation exhibits slightly inferior performance compared to the LS (oracle) solution, highlighting the negative impact of the randomness in the measurement matrix $\mathbf G$ on estimation accuracy. However, it significantly outperforms the LS method, demonstrating the positive effect of such randomness in providing additional information about the norm of the variable, i.e. $\|\mathbf{x}\|_2^2$, through the covariance structure of the measurements.

\subsection{Performance Comparisons}
In this subsection, we compare the performance of the estimators and the CRB derived in \cite{EldarCRB, sparseCRB} is also evaluated. We also implement TLS solution to make performance comparison \cite{SPLeldar}.

For the first experiment, the parameters are set as follows: $M=64$, $N=4$, $\sigma_e^2=0.10$, $\mathbf H$ and $\mathbf x$ are generated and fixed across 500 MC trials where $H_{ij}\sim {\mathcal N}(0,1)$ and $\mathbf x\sim {\mathcal N}(\mathbf 0,\mathbf I_N)$. The mean squared error (MSE) versus $10\log(1/\sigma_{{\epsilon}}^2)$ is shown in Fig. \ref{CRBverify}. It can be seen that there always exists a performance gap between the CRB and the LS estimator, which makes sense as the LS is unbiased and is mismatched to the model, and CRB acts as a lower bound for the LS. In contrast, the MLE is biased in Fig. \ref{biasverify}, and the MSE of the MLE is lower than the CRB.
\begin{figure*}
    \centering
    \subfigure[]{
    \label{CRBverify}
    \includegraphics[width = 2.4in]{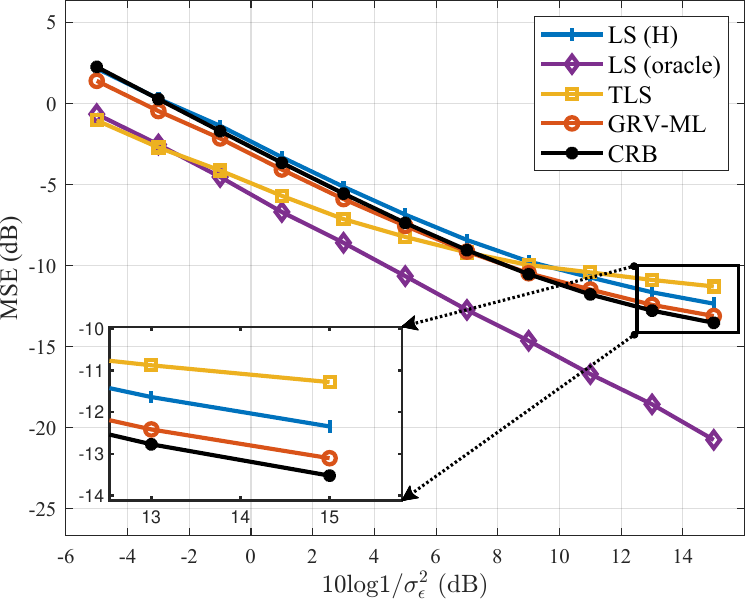}}
    % \subfigure[]{
    % \label{biasverifyLS}
    % \includegraphics[width = 2in]{biasverifyLS.pdf}}
    \subfigure[]{	
    \label{biasverify}
    \includegraphics[width = 2.5in]{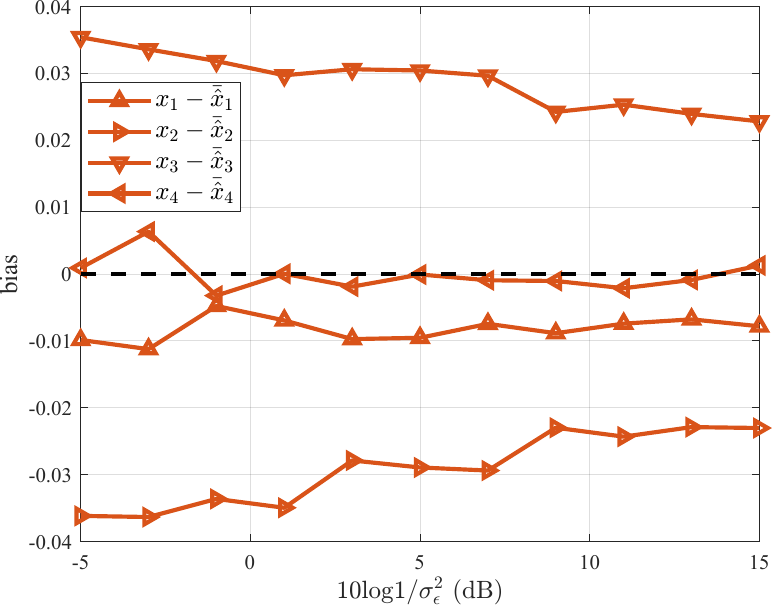}}
	\caption{(a): MSEs in estimating $\mathbf{x}$ of RV model.  
    % (b): bias of LS in estimating $\mathbf{x}$. 
    (b): bias of MLE in estimating $\mathbf{x}$. $\bar{\widehat x}_1$, $\bar{\widehat x}_2$, $\bar{\widehat x}_3$, $\bar{\widehat x}_4$ denote the average estimate of ${\widehat x}_1$, ${\widehat x}_2$, ${\widehat x}_3$, ${\widehat x}_4$, respectively. The bias of $\bar{\widehat x}_i$ is $\bar{\widehat x}_i-x_i$.}\label{fig2}
\end{figure*}

To further investigate the performance of the ML estimator, we consider a setting where the ratio $\kappa = \frac{\sigma^2_e}{\sigma^2_\epsilon}$ is held constant, and both $\sigma_e^2$ and $\sigma_\epsilon^2$ are scaled down proportionally. 
Three representative values of $\kappa$ are considered: $\kappa = 1$, $\kappa = 0.1$, and $\kappa = 0.01$. The results are shown in Fig.  \ref{kappaverify}. As observed, when $\kappa$ becomes smaller, the influence of the effective noise induced by $\sigma_e$ diminishes, and the ML estimator increasingly resembles the LS (oracle) solution. In particular, for sufficiently small $\kappa$, the ML estimator becomes asymptotically unbiased and converges to the LS (oracle) estimator.
\begin{figure*}
    \centering
    \subfigure[]{
    \label{kappa1}
    \includegraphics[width = 1.7in]{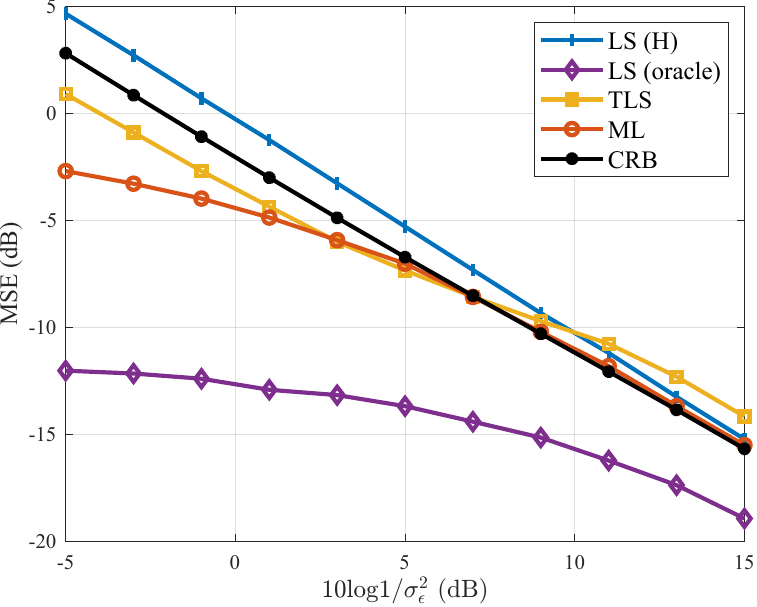}}
    \subfigure[]{
    \label{kappa01}
    \includegraphics[width = 1.7in]{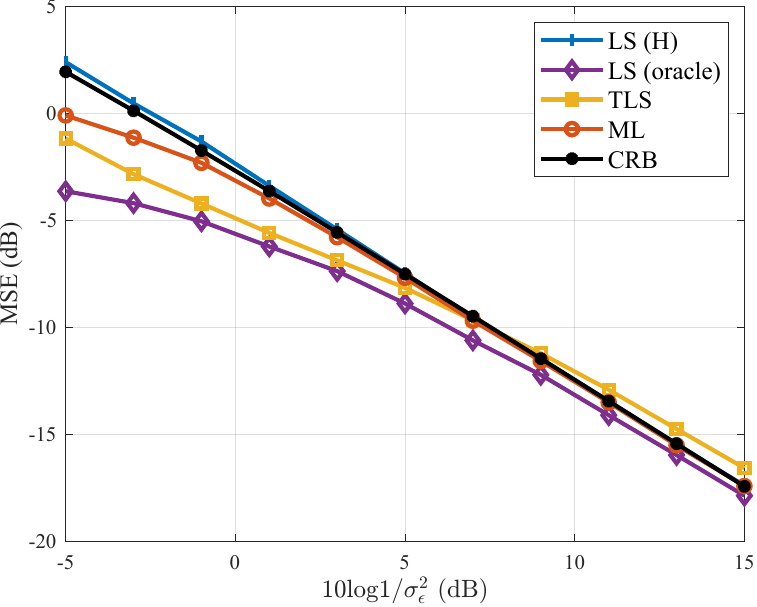}}
    \subfigure[]{	
    \label{kappa001}
    \includegraphics[width = 1.7in]{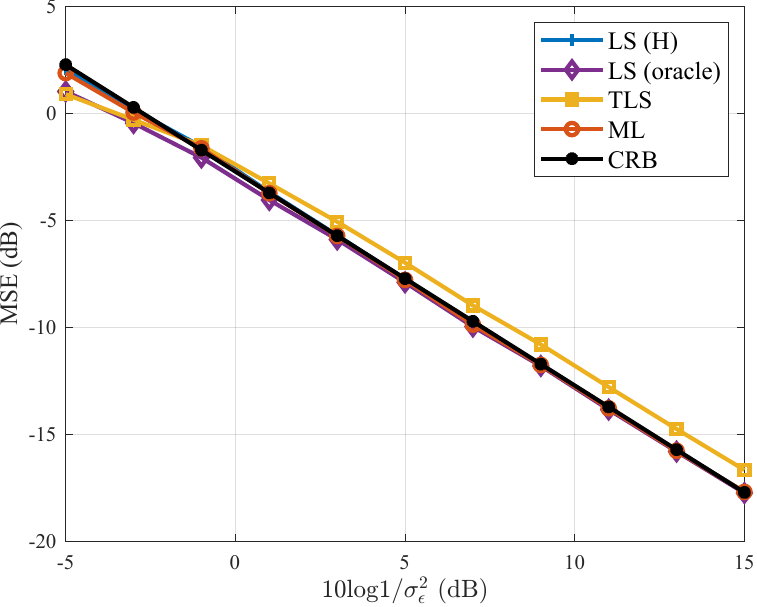}}
	\caption{MSEs in estimating $\mathbf{x}$ of random Gaussian model with different $\kappa=\frac{\sigma_e^2}{\sigma_\epsilon^2}$. (a): $\kappa=1$. (b): $\kappa=0.1$. 
    (c): $\kappa=0.01$. }\label{kappaverify}
\end{figure*}

Next, the asymptotic property of the MLE is examined. The parameters are set as follows: $N=4$, $\sigma_e^2=0.01$, $\sigma_{{\epsilon}}^2=0.20$, $\mathbf x$ is generated and fixed across the MC trials where $\mathbf x\sim {\mathcal N}(\mathbf 0,\mathbf I_N)$. The MSE versus the number of measurements $M$ is shown in Fig. \ref{mverify}. It can be seen that the MLE asymptotically approaches to the CRB, and performs better than the other methods without knowing the uncertainty.
\begin{figure}
	\centering
	\includegraphics[width = 3in]{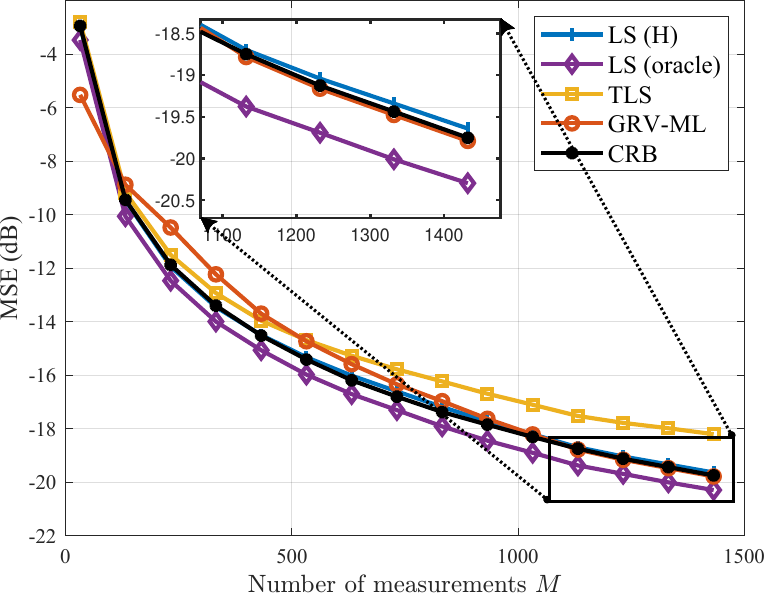}
	\caption{The MSE versus the number of measurements $M$.}\label{mverify}
\end{figure}

\section{Conclusion}
In this paper, we studied the linear regression with a Gaussian measurement model. By using the lifting technique, we transform the original optimization problem into a equivalent linear constrained optimization problem.
And it is proved that the equivalent MLE problem is convex and satisfies strong duality, strengthening previous quasi-convexity results.
By analyzing the KKT conditions, we derived a solution that accommodates all cases of rank-deficient and full-rank systems, as well as underdetermined and overdetermined configurations.
Compared to the previous solutions to this problem, this solution is more numerically efficient and general. Finally,  numerical results also demonstrate that the randomness in the measurement matrix has a negative effect in the estimation performance but also a positive one since it provides more information about the norm of the unknown vector through the covariance of the measurements under certain parameter settings.

\section{Acknowledgement}
This work was supported in part by the National Natural Science Foundation of China under Grant 62371420, in part by the National Natural Science Fund for Distinguished Young Scholars under Grant 62225114, in part by Zhejiang Provincial Natural Science Foundation of China under Grant LY22F010009, and in part by the National Natural Science Foundation of China under Grant 61901415.


\begin{thebibliography}{}
\bibitem{code}
``GRV-ML software,''  https://github.com/RiverZhu/MLERandomGaussianModel, 2025.

\bibitem{keyEst} 
S. M. Kay, \emph{Fundamentals of Statistical Signal Processing-Estimation
Theory.} Englewood Cliffs, NJ: Prentice-Hall, 1993.

\bibitem{raolin}
C. R. Rao, \emph{Linear Statistical Inference and its Applications.} New York, USA: Wiley, 1973.

\bibitem{errormodel}
W. A. Fuller, \emph{Measurement Error Models}, ser. Series in Probability and Mathematical Statistics. New York: Wiley, 1987.
\bibitem{TLS1}
G. H. Golub and C. F. Van Loan, ``An analysis of the total least-squares problem,'' \emph{SIAM J. Numer. Anal.}, vol. 17, no. 6, pp. 883–893, 1980.

\bibitem{TLS}
S. V. Huffel and J. Vanderwalle, ``The Total Least Squares Problem: Computational Aspects and Analysis,'' in \emph{Frontiers in Applied Mathematics 9}. Philadelphia, PA: SIAM, 1991.
\bibitem{robusteldar}
Y. C. Eldar, A. Ben-Tal and A. Nemirovski, ``Robust mean squared error estimation in the presence of model uncertainties,'' \emph{IEEE Trans. Signal Process.}, vol. 53, no. 1, pp. 168-181, Jan. 2005.
\bibitem{robustBeck}
A. Beck, A. Ben-Tal and Y. C. Eldar, ``Robust mean-squared error estimation of multiple signals in linear systems affected by model and noise uncertainties,'' \emph{Math. Programming}, vol. 107, pp. 155-187, Springer, 2006.
\bibitem{SPLeldar}
A. Wiesel, Y. C. Eldar and A. Beck, ``Maximum likelihood estimation in linear models with a Gaussian model matrix,'' \emph{IEEE Signal Process. Lett.}, vol. 13, no. 5, pp. 292-295, May 2006.

\bibitem{EldarCRB} 
A. Wiesel, Y. C. Eldar and A. Yeredor, ``Linear regression with Gaussian model uncertainty: Algorithms and bounds,'' \emph{IEEE Trans. Signal Process.}, vol. 56, no. 6, pp. 2194-2205, June 2008.

\bibitem{zhump}
J. Zhu, X. Lin, R. S. Blum, and Y. Gu, ``Parameter estimation from quantized observations in multiplicative noise environments,''\emph{IEEE Trans. Signal Process.}, vol. 63, no. 15, pp. 4037–4050, Aug. 2015.

\bibitem{ligang}
X. Wang, G. Li, and P. K. Varshney, ``Distributed detection of weak signals from one-bit measurements under observation model uncertainties,'' \emph{IEEE Signal Process. Lett.}, vol. 26, no. 3, pp. 415–419, Mar. 2019.

\bibitem{Sani}
A. Sani and A. Vosoughi, ``On distributed linear estimation with observation model uncertainties,'' \emph{IEEE Trans. Signal Process.}, vol. 66, no. 12, pp. 3212–3227, Jun. 2018.

\bibitem{ningzhang}
N. Jiang and N. Zhang, ``Expectation maximization-based target localization from range measurements in multiplicative noise environments,'' \emph{IEEE Commun. Lett.}, vol. 25, no. 5, pp. 1524–1528, May 2021.
 
\bibitem{multifading}
D. Ciuonzo, P. S. Rossi, and P. K. Varshney, ``Distributed detection in wireless sensor networks under multiplicative fading via generalized score tests,'' \emph{IEEE Internet Things J.}, vol. 8, no. 11, pp. 9059–9071, Jun. 2021.

\bibitem{SP}
S. Becker and R. J. Clancy, ``Robust least squares for quantized data matrices,'' \emph{Signal Process.}, vol. 176, Nov. 2020,
 Art. no. 107711.
\bibitem{CSboyd}
A. Zymnis, S. Boyd and E. Candes, ``Compressed sensing with quantized measurements,'' \emph{IEEE Signal Process. Lett.},  vol. 17, no. 2, pp. 149-152, Feb. 2010.

\bibitem{IRE}
B. Widrow, ``A study of rough amplitude quantization by means of Nyquist sampling theory,'' \emph{IRE Trans. Circuit Theory}, vol. 3, no. 4, pp. 266–276, 1956.

\bibitem{zhuMLE}
J. Zhu, X. Wang, X. Lin, and Y. Gu, ``Maximum likelihood estimation from sign measurements with sensing matrix perturbation,'' \emph{IEEE Trans. Signal Process.}, vol. 62, no. 15, pp. 3741–3753, Aug. 2014.
 
\bibitem{sparseCRB}
Y. Tang, L. Chen and Y. Gu, ``On the performance bound of sparse estimation with sensing matrix perturbation,'' \emph{IEEE Trans. Signal Process.}, vol. 61, no. 17, pp. 4372-4386, Sept.1, 2013.

\bibitem{cvxbook}
S. Boyd, L. Vandenberghe. \emph{Convex Optimization.} Cambridge University Press, 2004.

\bibitem{matrixcomp}
G. H. Golub, C. F. Van Loan. \emph{Matrix Computations.} The Johns Hopkins University Press, 2013.

\end{thebibliography}
\end{document}